\documentclass[12pt, reco]{article}
\usepackage{amssymb, amscd, amsmath, amsthm}
\usepackage[colorinlistoftodos]{todonotes}
\usepackage[colorlinks=true, allcolors=blue]{hyperref}
\usepackage{lmodern}
\usepackage{subcaption}
\usepackage{varwidth}

\usepackage[english]{babel}
\usepackage[rightcaption]{sidecap}
\usepackage{graphicx}
\usepackage{tikz}
\usepackage{caption}
\usepackage{subcaption}
\usepackage{pgfplots}
\usepackage{tikz-3dplot}

\usetikzlibrary{
    shapes,
    shapes.geometric,
    3d,
    perspective,
    positioning,
    intersections,
    arrows.meta,
    backgrounds,
    calc,
    decorations.pathreplacing
}

\usepackage[margin=1in]{geometry}

\newtheorem{theorem}{Theorem}[section]

\newtheorem{corollary}[theorem]{Corollary}
\newtheorem{lemma}[theorem]{Lemma}
\newtheorem{definition}[theorem]{Definition}
\newtheorem{example}[theorem]{Example}
\newtheorem{remark}[theorem]{Remark}

\def\Tr{\mathrm{Tr}}

\begin{document}

\title{Ergodic Theory of Inhomogeneous  Quantum Processes}

\maketitle

\centerline{ \author{ Abdessatar Souissi}}
\vskip0.25cm

\centerline{ Department of Management Information Systems, College of Business and Economics, }
\centerline{Qassim University, Buraydah 51452, Saudi Arabia}

\centerline{\textit{ a.souaissi@qu.edu.sa}}

\begin{abstract}
This work develops a rigorous framework for analysing ergodicity and mixing in time-inhomogeneous quantum dynamics. It considers quantum evolutions generated by sequences of quantum channels and examines in detail the relationship between the forward and backward dynamics, showing that they are generically nonequivalent in a structurally meaningful way. A central contribution is the adoption of a quantum Markov–Dobrushin approach to quantify mixing, which yields sharpened conditions for convergence rates and for establishing exponential stability of the induced dynamics. The resulting formalism not only extends classical and stationary quantum theories, but also naturally accommodates non-translationally invariant matrix product states, thereby providing a unified interface with experimentally relevant quantum many-body systems.
\end{abstract}




\section{Introduction}
Ergodic theory is a branch of mathematics that tells us how dynamical systems behave when we let them run for a very long time, and in particular when the average along a single trajectory matches the average taken over all possible states of the system. Starting from Boltzmann’s ergodic hypothesis~\cite{Bolzmann}, the basic idea is that an isolated system, if left to evolve long enough, will move through all microstates that are compatible with its energy, which makes it reasonable to replace time averages of physical quantities by averages over the whole phase space in statistical mechanics. The Ehrenfests later refined this picture with the quasi-ergodic hypothesis~\cite{Ehrenfest1907}, pointing out that a typical trajectory does not have to cover the energy surface uniformly, but only needs to come arbitrarily close to every point on it. Von Neumann and Birkhoff then gave these physical ideas a rigorous mathematical foundation through their ergodic theorems~\cite{vonN32,Birkoff31}, which show exactly how regular statistical behaviour can emerge from purely deterministic dynamics and clarify the conditions under which time and ensemble averages actually coincide.

   Markov chains have become one of the standard tools for describing random processes that evolve in time, especially when the future depends only on the current state and not on the full past history. Introduced by A.~A.~Markov in 1906~\cite{markov06}, they quickly proved useful for studying how such processes converge and mix, that is, how they forget their initial conditions as time goes on. Later work generalized the classical setting to non-homogeneous Markov chains with transition rules that can change in time, and showed that key ergodic features can still survive in this more flexible, time-dependent situation~\cite{D56,Haj58,Ib97}. These ideas naturally suggested a quantum analogue, obtained by replacing probability distributions with algebraic states and stochastic transition kernels with completely positive maps \cite{AccOh99}. This step led to the formulation of quantum Markov chains, first systematically developed by Accardi~\cite{Acc75}, which provide a rigorous framework for analysing how information propagates, how decoherence sets in, and how ergodic behaviour manifests itself in quantum systems. More recently, ergodic phenomena have been identified in quantum spin chains~\cite{SSB23,SM24} and in tree-like quantum structures~\cite{MG19}, showing that classical ergodic intuition can be meaningfully extended to non-commutative models.

In quantum information theory, the transmission and processing of information are most generally modeled by quantum channels \cite{gyongyosi2018survey, han2022quantum}, that is, completely positive, trace-preserving maps describing the evolution of quantum states in the presence of an environment. These maps provide a mathematically precise and physically faithful framework for capturing noise, decoherence, and interactions with external degrees of freedom. Beyond this structural role, concepts such as ergodicity and mixing have become central to understanding how quantum information propagates through a system, how it relaxes toward an effective equilibrium, and how it progressively loses memory of its initial preparation~\cite{bau2013,W12}. Recent developments have uncovered deep connections between these dynamical features and key information-theoretic quantities, including entropy production rates and a variety of divergence measures between quantum channels~\cite{fang2021geometric,gour2021entropy,MuWa2018}. Collectively, these results point to the emergence of robust statistical regularities under repeated application of quantum channels, indicating that the transport and dissipation of quantum information are governed by intrinsic ergodic mechanisms~\cite{aravinda2024ergodic,singh2024ergodic}.

A particularly interesting line of research shows that mixing properties play a crucial role in how reliably a quantum channel can transmit information, and in how strongly quantum data becomes scrambled under different dynamical regimes~\cite{ABM24,AFk23,singh2024zero}. At the same time, the study of structure-preserving quantum maps has provided important insights into the foundations of quantum error correction: such maps keep key operational quantities invariant during the evolution and thereby offer powerful strategies for encoding information in a way that remains robust against environmental noise and decoherence~\cite{raginsky2002strictly,blume2008characterizing,blume2010information, SA26}. In parallel, sequential generalized measurement schemes have emerged as versatile tools for describing complex quantum measurements, making it possible to simulate adaptive, non-projective measurement processes within realistic experimental architectures~\cite{ma2023sequential,linden2022arbitrary}.

From the perspective of quantum many-body physics, efficient asymptotic techniques based addressing the asymptotic behaviour of matrix product state (MPS) through  operator-algebraic methods have significantly sharpened our understanding  the long-time behavior   one-dimensional quantum systems  ~\cite{albert2019asymptotics,amato2023asymptotics}. A particularly important step forward in the analysis of non-stationary quantum dynamics was achieved by Movassagh and Schenker, who developed a general ergodic theory for temporally inhomogeneous quantum processes~\cite{movassagh2021theory}. Their framework provides a rigorous setting for studying ergodic and mixing properties in quantum evolutions driven by explicitly time-dependent generators, thereby helping to bridge the gap between traditional stationary treatments and genuinely time-dependent, driven quantum phenomena.

\subsection{Core Contributions}
This work develops a mathematical framework for  ergodic and mixing behaviour of time-inhomogeneous quantum processes, with the goal of rigorously characterizing the distinction of several hierarchy of  ergodic and mixing behaviours   in time-inhomogeneous dynamics and characterize  rates of convergence that arise in several processes. This extend both the previous studies in inhomogeneous classical processes in the framework of Markov chains and the homogeneous  quantum dynamics. The analysis is framed in terms of a sequence of quantum channels $\{\Phi_n\}_{n \in \mathbb{N}}$ acting on $\mathcal{B}(\mathcal{H})$, the algebra of bounded operators over a finite-dimensional Hilbert space $\mathcal{H}$ with focus on their action on the set of density operators $\mathfrak{S}(\mathcal{H})$  given its   versatile  role in the description of quantum systems. Each channel represents a step in the system’s temporal evolution, capturing the explicit inhomogeneity inherent to its dynamical behavior.

Two natural compositions arise in the context of inhomogeneous quantum dynamics: the \emph{forward} and \emph{backward} dynamics, defined respectively by
\[
\Phi_{m,n}^{(f)} := \Phi_n \circ \cdots \circ \Phi_{m+1} \circ \Phi_m,
\quad \text{and} \quad
\Phi_{m,n}^{(b)} := \Phi_m\circ \Phi_{m+1} \circ \cdots \circ \Phi_n,\quad m<n
\]
where $\circ$ denotes the composition of maps. These two constructions govern the temporal evolution of quantum states—represented as density operators belonging to the compact convex set $\mathfrak{S}(\mathcal{H})$—under the action of non-commuting, time-dependent quantum transformations. The divergence between the forward and backward evolutions, which typically emerges in the absence of commutativity among the channels, underscores the intrinsic temporal asymmetry in such systems. Consequently, these dynamics occupy a central role in the ergodic analysis of time-inhomogeneous quantum processes. Within this framework, ergodicity and mixing are understood not merely as asymptotic behaviors but as structural indicators of stability and convergence in long-term quantum evolution.

The following result establishes a fundamental asymmetry between forward and backward quantum dynamics—a distinction originating from a nesting property that manifests in the hierarchy of ergodic and mixing notions (Definition~\ref{definition:ergodicity_hierarchy}) and their weak counterparts (Definition~\ref{definition:weak_ergodicity_hierarchy}):

\begin{theorem} \label{theorem:characterizations_ergodicity_mixing}
Let $\{\Phi_n\}_{n \in \mathbb{N}}$ be a sequence of quantum channels on $\mathfrak{S}(\mathcal{H})$. Then the following statements hold:
\begin{enumerate}
    \item The backward process $\{\Phi^{(b)}_{1,n}\}$ is mixing (resp. exponentially mixing) if and only if it is weakly mixing (resp. exponentially weakly mixing).

    \item If the forward process $\{\Phi^{(f)}_{1,n}\}$ is mixing (resp. exponentially mixing), then it is weakly mixing (resp. exponentially weakly mixing). Conversely, if the images satisfy the nesting condition
    \begin{equation} \label{eq_nest}
        \Phi^{(f)}_{1,n+1}\bigl(\mathfrak{S}(\mathcal{H})\bigr)
        \subset
        \Phi^{(f)}_{1,n}\bigl(\mathfrak{S}(\mathcal{H})\bigr)
        \quad \text{for all } n
    \end{equation}
    Then the reverse implications also hold.
\end{enumerate}
\end{theorem}

While this nesting condition is automatically satisfied for the backward composition, it generally fails for the forward case. As a result, weak and strong mixing coincide in backward dynamics but may diverge in the forward scenario. This structural asymmetry is further elucidated through explicit counterexamples discussed below.

A methodological framework is developed based on the \emph{quantum Markov--Dobrushin inequality} introduced in \cite{AccLuSou22}, which provides a quantitative criterion for mixing of quantum channels and, more generally, non-stationary quantum dynamics. In this setting, the loss of distinguishability between states is measured by the quantum total variation norm \eqref{df-q-tot-var-nrm} and is controlled by a channel-dependent coefficient built from the action of the channel on rank-one (atomic) projections. Let
\(
   \mathcal{P}_1(\mathcal{H})
   :=
   \{\,|\xi\rangle\langle\xi|
     : \xi \in \mathcal{H},\ \|\xi\| = 1\,\}
\)
denote the set of rank-one orthogonal projections on \(\mathcal{H}\).

For a quantum channel \(\Phi : \mathcal{B}(\mathcal{H}) \to \mathcal{B}(\mathcal{H})\), the associated \emph{Markov--Dobrushin infimum set} is defined by
\begin{equation}\label{infMD}
   \mathcal{I}^{\Tr}_{\mathrm{MD}}(\Phi)
   :=
   \left\{
      \kappa_\Phi \in \mathcal{B}(\mathcal{H})
      \ \middle|
      \begin{aligned}
        &\centerdot \; 0 \le \kappa_\Phi \le \Phi(P), \, \forall\,P \in \mathcal{P}_1(\mathcal{H}),\\[0.25em]
        &\centerdot \; \operatorname{Tr}(\kappa_\Phi)
          \text{ is maximal among all } B \in \mathcal{B}(\mathcal{H}) \\[0.15em]
        &\text{\quad satisfying } 0 \le B \le \Phi(P),\,  \forall\,P \in \mathcal{P}_1(\mathcal{H})
      \end{aligned}
   \right\}
\end{equation}
The non-emptiness   of \( \mathcal{I}^{\Tr}_{\mathrm{MD}}(\Phi)\) is guaranteed by Lemma~\ref{Lem_MD}, and any element \(\kappa_\Phi \in  \mathcal{I}^{\Tr}_{\mathrm{MD}}(\Phi)\) is referred to as a \emph{Markov--Dobrushin infimum constant} of \(\Phi\). This variational construction generalizes the Markov--Dobrushin constant: instead of assuming the existence of a
genuine infimum (in the L\"owner order) of the family
\(\{\Phi(P) : P \in \mathcal{P}_1(\mathcal{H})\}\), which can be restrictive
and was left implicit in that earlier formulation, the present approach
selects a common lower bound with maximal trace among all such bounds. With this notation, the quantum Markov--Dobrushin inequality takes the form
\begin{equation}\label{MDineq}
  \Big\|\Phi(\rho) - \Phi(\sigma)\Big\|_{TV}
  \;\leq\;
  \bigl(1 - \operatorname{Tr}(\kappa_\Phi)\bigr)\,\Big\|\rho - \sigma\Big\|_{TV},
  \qquad \rho,\sigma \in \mathfrak{S}(\mathcal{H})
\end{equation}
where \(\kappa_\Phi \in  \mathcal{I}^{\Tr}_{\mathrm{MD}}(\Phi)\).

  The maximality-of-trace condition in \eqref{infMD} ensures that the coefficient \(\eta_{\mathrm{MD}}(\Phi) := 1 - \operatorname{Tr}(\kappa_\Phi)\) is as small as possible among all admissible common lower bounds, and hence yields a \emph{minimal} Markov–Dobrushin contraction factor for the channel in \eqref{MDineq}. In this sense, \(\eta_{\mathrm{MD}}(\Phi) \) plays the role of a channel-dependent contraction coefficient: it captures, in a single scalar, how much \(\Phi\) can shrink the quantum total variation distance between states and thus how strongly it suppresses distinguishability in the Markov–Dobrushin sense. More precisely, \(\eta_{\mathrm{MD}}(\Phi)\) provides an explicit \emph{upper bound} on the optimal contraction coefficient of \(\Phi\) with respect to the trace distance, that is,
\[
\eta_{\mathrm{Tr}}(\Phi)
:= \sup_{\rho\neq\sigma}
\frac{\|\Phi(\rho)-\Phi(\sigma)\|_1}{\|\rho-\sigma\|_1}
\;\leq\; \eta_{\mathrm{MD}}(\Phi),
\]
so that the Markov–Dobrushin constant controls, in a computable way, the worst-case distinguishability loss between quantum states under the action of \(\Phi\). In this sense, \(\eta_{\mathrm{MD}}(\Phi)\) can be viewed as a trace-norm (or trace-distance) Dobrushin coefficient for the channel, consistent with recent developments on quantum Doeblin coefficients and sharp contraction bounds for quantum channels~\cite{George2025,Hirche2024}. In particular, while the exact trace-norm contraction coefficient is known to be computationally intractable—NP-hard to approximate within constant factors in closely related formulations \cite{Delsol2025}—the Markov–Dobrushin infimum constant furnishes a practically accessible surrogate: \(\eta_{\mathrm{MD}}(\Phi)\) is easier to compute or bound in practice, yet still reflects the essential contraction behaviour of the channel. All subsequent ergodic and mixing results remain valid if one replaces \(\eta_{\mathrm{MD}}(\Phi)\) by any admissible trace-norm contraction coefficient, at the price of working with a potentially sharper but computationally less tractable constant.

Building on the foundational Markov--Dobrushin inequality, the following theorem provides a broad convergence criterion for non-homogeneous quantum dynamics. In particular, it guarantees asymptotic stability for both forward and backward evolutions without imposing a temporally uniform contractivity condition, marking a substantial departure from classical treatments~\cite{Muk15,GQ15,Sa15}.

\begin{theorem}\label{thm-main_MD_mixing}
Let $\{\Phi_n\}_{n \in \mathbb{N}}$ be a sequence of quantum channels, and let $ \kappa_{\Phi_n}\in  \mathcal{I}^{\Tr}_{\mathrm{MD}}(\Phi_n)$ denote a corresponding sequence of Markov--Dobrushin coefficients. Suppose that $\{\kappa_{\Phi_n}\}$ admits a non-zero accumulation point. Then there exist $\mu \in (0,1)$ and an increasing unbounded sequence $\{N(n)\}_{n \in \mathbb{N}}$ such that, for all $t \in \{b,f\}$,
\begin{equation}\label{eq_exp_conv}
   \forall m\in\mathbb{N}, \qquad  \sup_{\rho,\sigma \in \mathfrak{S}(\mathcal{H})}
    \Big\|
        \Phi^{(t)}_{m+1,n}(\rho) - \Phi^{(t)}_{m,n}(\sigma)
    \Big\|_{TV}
    \leq 2 \mu^{N(n)- N(m)}
\end{equation}
As a consequence, both the forward process $\{\Phi^{(f)}_{m+1,n}\}$ and the backward process $\{\Phi^{(b)}_{m+1,n}\}$ are weakly mixing, with state differences decaying at least exponentially fast in $N(n)$.
\end{theorem}

This theorem furnishes a unified and flexible framework for quantifying convergence rates in inhomogeneous quantum regimes. When $N(n) = \mathcal{O}(n)$, one recovers genuinely exponential mixing, whereas the scaling $N(n) = \mathcal{O}(\ln n)$ leads to sub-exponential, for instance polynomial, decay. The effective rate of mixing is thus dictated by the temporal distribution of strongly contractive steps, highlighting the intricate interplay between local contraction mechanisms and the resulting global dynamical behavior.

\subsubsection{Application to Inhomogeneous Matrix Product States}
To substantiate the practical significance of the developed ergodic framework, we examine its application to the class of \emph{non-translation-invariant matrix product states} (MPS), which constitute paradigmatic models of inhomogeneous quantum spin chains \cite{perez2007matrix, verstraete2008matrix, Chen24}. These states are naturally defined on the quasi-local $C^*$-algebra
\[
\mathfrak{A}_{\mathbb{N}} := \bigotimes_{n \in \mathbb{N}} \mathcal{B}(\mathcal{K})
\]
where each local algebra $\mathcal{B}(\mathcal{K})$ denotes the bounded operators on a finite-dimensional Hilbert space $\mathcal{K} \cong \mathbb{C}^m$.

The MPS structure is generated by a site-dependent family of Kraus operators $\{K_i^{[n]}\}_{i=1}^m \subset \mathcal{B}(\mathcal{H})$ acting on an auxiliary Hilbert space $\mathcal{H} \cong \mathbb{C}^d$.  The corresponding $n$-site MPS vector is given by
\begin{equation}\label{eqMPS}
|\Psi_n\rangle
    = \sum_{i_1=1}^m \cdots \sum_{i_n=1}^m
      \mathrm{Tr}\!\left(K_{i_1}^{[1]} K_{i_2}^{[2]} \cdots K_{i_n}^{[n]}\right)
      |i_1 i_2 \cdots i_n\rangle
    \;\in\; \mathcal{C}^{\otimes n}
\end{equation}

In order to relate the MPS data $\{K_i^{[n]}\}$ to the ergodic properties of the channel sequence $\{\Phi_n\}$, it is convenient to work with canonical forms. We say that the MPS is in \emph{left-canonical form} if, for each site $n$,
\begin{equation}\label{lcform}
\sum_{i=1}^m \bigl(K_i^{[n]}\bigr)^\dagger K_i^{[n]} \;=\; \mathbb{I}
\end{equation}
so that the map
\begin{equation}\label{PhinMPS}
\Phi_n : \mathcal{B}(\mathcal{H}) \to \mathcal{B}(\mathcal{H})
\qquad
\Phi_n(X) := \sum_{i=1}^m K_i^{[n]} X \bigl(K_i^{[n]}\bigr)^\dagger
\end{equation}
is completely positive and trace-preserving, then a quantum channel acting on density operators on $\mathcal{H}$. The tensors $K_i^{[n]}$ simultaneously encode the MPS structure on the spin chain and the inhomogeneous channel dynamics on the auxiliary space.

Similarly, a \emph{right-canonical form} is obtained by requiring
\( \sum_{i=1}^m K_i^{[n]} \bigl(K_i^{[n]}\bigr)^\dagger \;=\; \mathbb{I},\)
which corresponds to unitality of the associated channel. In the non-translation-invariant setting considered here, one may mix these normalizations along the chain, but we shall assume that a suitable left-canonical gauge has been chosen so that each $\Phi_n$ is a quantum channel on $\mathfrak{S}(\mathcal{H})$ and the norm of $|\Psi_n\rangle$ is controlled uniformly in $n$.

This construction defines a sequence of normalized states $\{\varphi_n\}$ on the local algebras
\[
\mathfrak{A}_{[0,n]} := \bigotimes_{k=0}^n \mathcal{B}(\mathcal{K})
\]
via
\[
\varphi_n(X)
    := \frac{\langle \Psi_n | X | \Psi_n \rangle}{\langle \Psi_n | \Psi_n \rangle},
    \qquad X \in \mathfrak{A}_{[0,n]}
\]
In this way, the time-inhomogeneous channel sequence $\{\Phi_n\}$ arising from the MPS tensors $\{K_i^{[n]}\}$ fits directly into our ergodic framework, and the Markov--Dobrushin analysis developed in Theorem~\ref{thm-main_MD_mixing} can be used to characterize time-uniform mixing and loss of memory for non-translation-invariant MPS.

Leveraging the convergence framework established in Theorem~\ref{thm-main_MD_mixing}, we now apply these results to inhomogeneous matrix product states. Each step in the sequence is governed by a quantum channel $\Phi_n$ on $\mathcal{B}(\mathcal{H})$, induced by the Kraus operators $\{K_i^{[n]}\}_{i=1}^m$ that encode the evolving correlations in the auxiliary space of the MPS.
\begin{theorem}[Convergence of inhomogeneous matrix product states]\label{thm_main_inhom_MPS}
Let $\{K_i^{[n]}\}_{i=1}^m \subset \mathcal{B}(\mathcal{H})$ be the site-dependent tensors of a non-translation-invariant MPS \eqref{eqMPS}, chosen in left-canonical form in the sense of \eqref{lcform}. Let $\{\Phi_n\}_{n\in\mathbb{N}}$ be the associated sequence of quantum channels \eqref{PhinMPS} on $\mathcal{B}(\mathcal{H})$, and let $\{\kappa_{\Phi_n}\}$ denote the corresponding Markov--Dobrushin coefficients. Assume that $\{\kappa_{\Phi_n}\}$ admits a non-zero accumulation point. Then there exists a unique state $\varphi_{\infty}$ on the quasi-local algebra $\mathfrak{A}_{\mathbb{N}}$ such that
\[
\varphi_n \xrightarrow[n\to\infty]{\text{weak-*}} \varphi_{\infty}
\]
Moreover, there exists a sequence of backward boundary conditions $(\rho_m^{(b)})_{m\in\mathbb{N}} \subset \mathfrak{S}(\mathcal{H})$ satisfying
\[
\rho_m^{(b)} = \Phi_{m}(\rho_{m+1}^{(b)}), \qquad m\in\mathbb{N}
\]
such that:
\begin{equation}\label{eq_phiinfty_short}
    \varphi_{\infty}(X)
        = \sum_{\substack{i_1,\ldots,i_m\\ j_1,\ldots,j_m}}
     \langle i_1,\ldots,i_m|X|j_1,\ldots,j_m\rangle\,
    \Tr\Big( K^{[1]}_{j_1}\cdots K^{[m]}_{j_m}\, \rho_{m+1}^{(b)}\,
     (K^{[m]}_{i_m})^\dagger\cdots (K^{[1]}_{i_1})^\dagger\Big)
\end{equation}
for every local observable $X \in \mathfrak{A}_{[0,m]}$.
\end{theorem}

Theorem~\ref{thm_main_inhom_MPS} furnishes a concrete realization of the abstract mixing conditions in Theorem~\ref{thm-main_MD_mixing} within the class of non-translation-invariant matrix product states. It provides a rigorous and computationally tractable framework for describing the ergodic behaviour of quantum spin chains with spatially inhomogeneous interactions: a purely dynamical assumption on the associated channel sequence \(\{\Phi_n\}_n\), expressed via the Markov–Dobrushin coefficients \(\{\kappa_{\Phi_n}\}_n\), is translated into the existence and uniqueness of a thermodynamic limit state \(\varphi_\infty\) together with the explicit representation \eqref{eq_phiinfty_short} of its local expectations in terms of the MPS tensors and a sequence of boundary conditions.

The formula \eqref{eq_phiinfty_short} also illustrates the structural role of the backward boundary states \((\rho_m^{(b)})_{m\in\mathbb{N}}\): they encode, for each cut \(m\), the effective environment “seen from the right” by the first \(m\) sites and arise canonically from the time-uniform mixing of the backward dynamics. In this way the theorem makes precise how the auxiliary-space dynamics of an inhomogeneous MPS determines its infinite-volume limit.

\subsection{ Related Studies}

The study of ergodic properties in quantum channels and processes has seen significant developments in recent years. A particularly influential framework analyzes the ergodic behavior of time-inhomogeneous quantum dynamics through random quantum channels~\cite{movassagh2022a, movassagh2021theory}. This approach characterizes ergodicity via an ergodic sequence  of density matrices under an irreducibility condition, referred in the present work as \textit{Trajectory-Ergodicity}, grounded on the extension of Perron--Frobenius theory to  Banach algebras through the Krein--Rutman theorem~\cite{KR50}. This body of work collectively clarifies the rich structural interplay between various manifestations of ergodicity and mixing in quantum channels, and its implications for quantum dynamics and information processing.

Recent advances in quantum information theory have further sharpened this picture by linking dynamical behavior to concrete operational tasks, such as quantum channel learning~\cite{belov2025quantum} and the detectability of quantum channel capacities~\cite{singh2022detecting}, where ergodic and mixing properties constrain achievable performance in noisy or time-varying settings. Fundamental limits on resource distillation in general quantum resource theories~\cite{regula2021fundamental} provide an additional incentive to investigate ergodic phenomena, as they reveal how long-term dynamical behavior governs the asymptotic conversion of nonclassical resources. Geometric and combinatorial approaches, notably those based on the structure of Birkhoff’s polytope and the characterization of unistochastic matrices~\cite{ben2005}, offer complementary insight into the constraints shaping admissible quantum dynamics and their mixing profiles. A further compelling direction concerns the relationship between the mixing behavior of quantum channels and that of states generated by quantum Markov chains~\cite{ASS20}, including entangled Markov chains~\cite{SSB23}, where non-commutative correlations complicate the passage from channel-level to state-level ergodicity. This line of inquiry is closely tied to the broader problem of reconstructing quantum Markov chains from suitable families of quantum channels or their duals~\cite{AccOh99}, thereby linking dynamical features to algebraic ergodic properties within the framework of $C^*$-dynamical systems~\cite{Fid09,Fid10}.

Random quantum channels provide a natural and versatile framework for modeling noise and information flow in fluctuating quantum systems~\cite{CN10,CN11II}, and, when augmented with physically motivated symmetries, they offer an even more faithful description of experimentally relevant scenarios~\cite{NP25}. Recent work by Movassagh and Schenker has established a rigorous ergodic theory for inhomogeneous and explicitly time-dependent quantum processes~\cite{movassagh2022a,movassagh2021theory}, with notable applications to the thermodynamic behavior of matrix product states (MPS)~\cite{perez2007matrix,verstraete2008matrix}. This line of research reveals a profound link between quantum dynamical complexity and random matrix theory. In particular, random MPS models~\cite{GOZ10} introduce classical randomness in a controlled manner, thereby enhancing analytical tractability while preserving the essential quantum features of the states under consideration~\cite{Aci}. Against this backdrop, a natural question arises as to whether the Markov--Dobrushin condition employed here is fundamentally equivalent to the irreducibility assumptions underlying the forward dynamics in~\cite{movassagh2022a}. Establishing such an equivalence would bridge two complementary perspectives and contribute to a more unified theory of ergodicity and mixing in quantum dynamics. Moreover, our results suggest that a genuinely probabilistic viewpoint, based on ensembles of random quantum channels, may shed additional light on the Markov--Dobrushin mixing mechanism. This perspective points toward a promising research direction in which refined probabilistic tools could be developed to analyze the robustness and fluctuation properties of ergodic and mixing behavior for random quantum channels~\cite{NP12}.

The structure of the paper is as follows. In Section~\ref{sect_preliminaries}, we introduce the necessary background and mathematical preliminaries, including key notions from quantum probability and inhomogeneous quantum dynamics. Section~\ref{sect_inhomogeneousQP} formulates the central concepts of ergodicity and mixing for time-inhomogeneous quantum processes, thereby laying the conceptual foundation for the rest of the work. In Section~\ref{sect_Mixcharacterization}, we investigate several mixing properties in detail and provide rigorous characterizations of their behavior in the quantum setting. Section~\ref{sect_MDmixing} is devoted to extending the Markov--Dobrushin mixing condition to inhomogeneous quantum evolutions, emphasizing how this framework differentiates between distinct degrees of mixing. Finally, in Section~\ref{sect_AppMPS}, we apply the developed theory to non-translation-invariant matrix product states, illustrating the practical implications and effectiveness of our results in the analysis of inhomogeneous quantum spin chains.

\section{Preliminaries}\label{sect_preliminaries}

Let \( d \in \mathbb{N} \) and let \( \mathcal{H} \) denote a \( d \)-dimensional Hilbert space, equipped with a fixed orthonormal basis \(\{ | i \rangle \}_{i=1}^d\). The algebra of all bounded linear operators on \( \mathcal{H} \), denoted by \( \mathcal{B}(\mathcal{H}) \), is a unital Banach algebra in which the identity operator \( \mathbb{I} \) serves as the multiplicative unit. This algebra is canonically endowed with the  trace norm, defined for any operator \( A \in \mathcal{B}(\mathcal{H}) \) by
\[
\| A \|_{1} := \operatorname{Tr} \big( \sqrt{A^* A} \big) = \sum_{i=1}^{d} \sigma_i(A)
\] where \( \sigma_i(A) \) are the singular values of \( A \), i.e., the eigenvalues of \( \sqrt{A^* A} \) counted with multiplicity, and \( \operatorname{Tr}(\cdot) \) denotes the usual trace on \( \mathcal{H} \) relative to the given orthonormal basis. The trace norm \( \|\cdot\|_{1} \) will serve as our primary notion of size and distance for operators.In finite dimensions, all standard operator norms are equivalent, but the trace norm is distinguished in quantum information theory because, for density operators, the trace distance $\|\rho - \sigma\|_{1}$ has a direct operational interpretation as a measure of state distinguishability and is naturally compatible with the action of completely positive, trace-preserving maps modeling quantum channels~\cite{Delsol2025,Hirche2024}. Throughout this work, all quantitative statements on convergence, mixing, and ergodic behavior of quantum processes will therefore be expressed in terms of the trace norm.
 Recall that a state  on $\mathcal{B}(\mathcal{H})$ is a positive linear map $\varphi$  such that  $\varphi(\mathbb{I}) = 1$.  We denote $\mathcal{S}(\mathcal{H})$ the  set of states on $\mathcal{B}(\mathcal{H})$.
The set of density operators on \( \mathcal{H} \), which describe quantum states, is defined as:
\[
\mathfrak{S}(\mathcal{H}) = \left\{ \rho \in \mathcal{B}(\mathcal{H}) \mid \rho = \rho^*, \rho \geq 0, \operatorname{Tr}(\rho) = 1 \right\}
\]

Here, \( \rho^* \) is the adjoint of \( \rho \), and the trace condition \( \operatorname{Tr}(\rho) = 1 \) ensures normalization. This set is convex and compact---key properties that make it indispensable in quantum information theory, especially for optimization tasks and variational approaches.
The map $\rho \mapsto \Tr(\rho \,\cdot)$ define a bijective map from $\mathfrak{S}(\mathcal{H})$ onto $\mathcal{S}(\mathcal{H})$.
For a given probability measure $\mu\in\mathcal{P}_d$ the expectation value on maps $f: \{1,\dots, d\} \to \mathbb{C}$ with respect to the measure $\mu$
\begin{equation}\label{eq_mu(f)}
\mu(f) = \sum_{i} f(i)\mu(\{i\})
\end{equation}
defines a state on  the diagonal subalgebra $\mathcal{D} $ spanned by the brackets $|i\rangle\langle i|$. The correspondent density operator is $\rho_{\mu} = \sum_{i}\mu(\{i\})|i\rangle\langle i|$ allows the extension of  (\ref{eq_mu(f)}) to define a state on  the full algebra $\mathcal{B}(\mathcal{H})$ through:
$$
\varphi_{\mu} : A \mapsto \Tr(\rho_{\mu}A)
$$

Within $\mathcal{B}(\mathcal{H})$, the positive cone $\mathcal{B}(\mathcal{H})_+$ is defined as the set of all positive semidefinite operators acting on $\mathcal{H}$. These operators occupy a central position in quantum theory, as they represent physically admissible quantities such as density operators describing quantum states and effects associated with measurement outcomes. Formally, an operator $A \in \mathcal{B}(\mathcal{H})$ belongs to $\mathcal{B}(\mathcal{H})_+$ if and only if
\[
\langle \psi, A \psi \rangle \geq 0 \quad \forall\, \psi \in \mathcal{H}
\]
which expresses the requirement that all expectation values of $A$ are non-negative in every vector state. This positivity condition underpins the mathematical formulation of valid quantum states, measurement procedures, and generalized measurement schemes, including positive operator-valued measures (POVMs) that extend the standard projective measurement framework.

\begin{theorem}[Kraus Representation]\label{thm:kraus_rep}\cite{Kraus71, Choi75}
A linear map $\Phi: \mathcal{B}(\mathcal{H}) \to \mathcal{B}(\mathcal{H})$ is a quantum channel if and only if it admits a  Kraus representation
\[
\Phi(A) = \sum_{i=1}^m K_i A K_i^\dagger
\]
for every $A \in \mathcal{B}(\mathcal{H})$, where the  Kraus operators  $\{K_i\}_{i=1}^m \subset \mathcal{B}(\mathcal{H})$ satisfy
\[
\sum_{i=1}^m K_i^\dagger K_i = \mathbb{I}
\]
\end{theorem}

\begin{remark}
The set of Kraus operators is not unique; two such sets $\{K_i\}_{i=1}^m$ and $\{L_j\}_{j=1}^n$ represent the same channel if and only if there exists an $m \times n$ complex matrix $U = (u_{ij})$ with $U^\dagger U = \mathbb{I}$ such that $K_i = \sum_{j=1}^n u_{ij} L_j$ for all $i$. If $\Phi$ is unital, i.e., $\Phi(\mathbb{I})=\mathbb{I}$, then additionally $\sum_{i=1}^m K_i K_i^\dagger = \mathbb{I}$.

The Kraus (or operator-sum) representation is the mathematical realization of a quantum channel as an open-system evolution: it captures the effect of the environment on the system via the operators $K_i$, while the trace-preserving condition $\sum_i K_i^\dagger K_i = \mathbb{I}$ guarantees that the transformation maps density operators to density operators.
\end{remark}
\begin{definition}
A quantum channel is defined as a linear map $\Phi$ between operator spaces $\mathcal{B}(\mathcal{K})$ and $\mathcal{B}(\mathcal{K}')$ that maintains two crucial properties: complete positivity and trace preservation (CPTP). These requirements guarantee the physical implementability of such transformations. In our current framework, we specifically examine channels where both input and output spaces coincide, i.e., $\Phi : \mathcal{B}(\mathcal{H}) \to \mathcal{B}(\mathcal{H})$, with $\mathcal{H}$ being our previously defined finite-dimensional Hilbert space.\\
Let $\Phi^* : \mathcal{B}(\mathcal{H}) \to \mathcal{B}(\mathcal{H})$ be the Hilbert-Schmidt dual of a quantum channel $\Phi$   in  the following sense
 $$\langle X, \Phi(Y)\rangle_{\text{HS}} = \langle \Phi^*(X), Y\rangle_{\text{HS}}\qquad \forall X,Y \in \mathcal{B}(\mathcal{H})
 $$
  We call $\Phi^*$ a \emph{Markov operator} since it is completely positive and unital.
 These properties identify $\Phi^*$ as the natural quantum counterpart of a classical Markov transition matrix in the framework of quantum Markov chains \cite{ASS20}: complete positivity generalizes the entrywise positivity of classical transition probabilities, and preservation of the identity operator implements the analogue of total probability conservation.
\end{definition}

\begin{remark}
The duality induced by the Hilbert--Schmidt inner product is not merely a formal adjunction, but rather encodes a substantive correspondence between the Schrödinger and Heisenberg pictures of quantum dynamics. More precisely:
\begin{itemize}
    \item Trace preservation of $\Phi$ is equivalent to unitality of $\Phi^*$.
    \item Complete positivity of $\Phi$ is equivalent to complete positivity of $\Phi^*$.
\end{itemize}
Accordingly, $\Phi$ describes the evolution of states (density operators), whereas $\Phi^*$ governs the evolution of observables. The Hilbert--Schmidt structure is crucial in this context, as it endows $\mathcal{B}(\mathcal{H})$ with a canonical Hilbert space geometry, renders the adjoint map well defined, and permits the systematic application of spectral and ergodic techniques to quantum dynamical systems.
\end{remark}

The iterated application of a quantum channel $\Phi$ is denoted by its $n$-fold composition:
$$
\Phi^{n} = \underbrace{\Phi \circ \Phi \circ \cdots \circ \Phi}_{n \text{ times}}
$$
with the convention that $\Phi^{0}$ equals the identity superoperator $\mathcal{I}$ on $\mathcal{B}(\mathcal{H})$.

For any completely positive linear map $\Phi : \mathcal{B}(\mathcal{H}) \rightarrow \mathcal{B}(\mathcal{K})$, the operator-sum representation takes the form:
\begin{equation*}
  \Phi(\rho) = \sum_{i} K_i \rho K_i^\dagger
\end{equation*}
where $\{K_i\}$ are the associated Kraus operators. An essential feature of trace-preserving maps is characterized by the normalization condition:
\begin{equation*}
    \sum_{i} K_i^\dagger K_i = \mathbb{I}
\end{equation*}
which ensures probability conservation in quantum measurements and transformations.

A quantum channel $\Phi$ is called \emph{ergodic} if, for any initial state $\rho$, its time-averaged evolution converges to a unique fixed point $\rho_*$ satisfying $\Phi(\rho_*) = \rho_*$:

$$
     \lim_{n \to \infty} \frac{1}{n+1} \sum_{k=0}^{n} \Phi^k(\rho) = \rho_*
$$
The limiting channel  extends into a super-operator on  $\mathcal{B}(\mathcal{H})$ given by
\begin{equation}\label{eq_erg}
\widehat{\Phi}(M) := \lim_{n \to \infty} \frac{1}{n+1}\sum_{k=0}^n \Phi^k(M) = \Tr(M)\rho_*
\end{equation}
 is defined on the full  algebra,  it is idempotent ($\widehat{\Phi} \circ \widehat{\Phi} = \widehat{\Phi}$) and it  satisfies $\widehat{\Phi} \circ \Phi = \widehat{\Phi}$. Moreover, it projects any state onto $\rho_*$, reflecting the long-time loss of initial state information.

A stronger property, \emph{mixing}, requires the convergence of individual trajectories for all pairs of states:
\begin{equation}\label{eq_df_mixing}
    \lim_{n \to \infty} \Big\|\Phi^n(\rho) - \Phi^n(\sigma)\Big\|_1 = 0 \quad \forall \rho,\sigma \in \mathfrak{S}(\mathcal{H})
\end{equation}
While all mixing channels are ergodic, the converse fails: ergodicity depends solely on the fixed-point structure, whereas mixing additionally demands specific spectral properties of $\Phi$. The distinction highlights that mixing enforces uniformity in long-time behavior, while ergodicity only guarantees convergence in the time-averaged sense \cite{bau2013, brasil2021lyapunov}.

\section{Ergodicity and Mixing for Inhomogeneous Quantum Processes}\label{sect_inhomogeneousQP}

In this section, we develop a framework for defining ergodicity and mixing for quantum systems that evolve in time under a sequence of distinct quantum channels. Our perspective is informed by ergodic theory and  the classical theory of inhomogeneous Markov chains, in which analogous questions concerning the asymptotic behavior of time-dependent stochastic processes have been investigated extensively over several decades~\cite{CFS82,W2000,Moore15,Breamaud20}.

In the homogeneous quantum case, one iterates a single transition rule (or quantum channel), whereas in the inhomogeneous case the dynamics is driven by a sequence of time-dependent kernels, and ergodic properties must be understood at the level of the whole family rather than a single operator. For classical inhomogeneous Markov chains, this has led to a refined terminology and a rich toolkit based on Dobrushin-type coefficients \cite{D56}, which distinguish weak and uniform ergodicity and quantify loss of memory of initial conditions. In the next subsection, we recall these classical notions of ergodicity and mixing for inhomogeneous Markov chains as they appear in the existing literature, both to fix terminology and to indicate which structural features have meaningful analogues—and which do not—in the  inhomogeneous quantum setting.

\subsection{ Ergodicity of Inhomogeneous Markov Chains}

Consider a discrete state space $S = \{1,\ldots,d\}$ with $d \geq 2$. Let $\mathcal{P}_d$ denote the probability simplex of all probability measures on $S$, represented as row vectors in $\mathbb{R}^d$ with non-negative components summing to one. The time evolution is governed by a sequence of right-stochastic matrices $\{\Pi_n\}_{n\geq 0} \subset \mathbb{R}^{d\times d}$, where each $\Pi_n$ satisfies $(\Pi_n)_{ij} \geq 0$ and $\sum_{j=1}^d (\Pi_n)_{ij} = 1$ for every $i \in S$~\cite{Breamaud20}. Thus, for each $n$, the map $\mu \mapsto \mu \Pi_n$ sends $\mathcal{P}_d$ into itself.

For integers $0 \le m < k$, the multi-step transition operator from time $m$ to $k$ is given by the ordered matrix product
\begin{equation}\label{eq:matrix_composition}
    \Pi^{(m,k)} := \Pi_{m}\Pi_{m+1}\cdots\Pi_{k-1} \in \mathbb{R}^{d\times d}
\end{equation}
which remains stochastic. Hence, for any initial distribution $\mu_0 \in \mathcal{P}_d$, the law at time $k$ is
\begin{equation}\label{eq:distribution_evolution}
    \mu_k = \mu_0 \Pi^{(0,k)} \in \mathcal{P}_d
\end{equation}

\medskip
On $\mathcal{P}_d$ we measure distances using the total variation norm, which in the finite setting coincides (up to the constant $1/2$) with the $\ell^1$-norm:
\[
\big\|\mu - \nu\big\|_{\mathrm{tv}}
    := \frac{1}{2}\sum_{i=1}^d |\mu(i) - \nu(i)|,
    \qquad \mu,\nu \in \mathcal{P}_d
\]
The quantity $\big\|\mu\Pi^{(m,k)} - \nu\Pi^{(m,k)}\big\|_{\mathrm{tv}}$ thus measures how far apart the laws at time $k$ are when starting from two different initial distributions at time $m$.

The Markov--Dobrushin coefficient of a stochastic matrix $\Pi$ is defined by
\begin{equation}\label{eq:MD_TV}
    \delta(\Pi)
    := \sup_{x_1,x_2 \in S}
       \bigl\|\Pi(x_1,\cdot) - \Pi(x_2,\cdot)\bigr\|_{\mathrm{tv}}
\end{equation}
Equivalently, $\delta(\Pi)$ is the smallest constant such that
\[
\big\|\mu \Pi - \nu \Pi\big\|_{\mathrm{tv}}
    \le \delta(\Pi)\,\big\|\mu - \nu\big\|_{\mathrm{tv}}
    \quad \text{for all } \mu,\nu \in \mathcal{P}_d
\]
For an inhomogeneous Markov chain $\{\Pi_n\}_{n\ge 0}$, the sequence $\{\delta(\Pi_n)\}_{n\ge 0}$ therefore quantifies the stepwise contraction in total variation and will be used to formulate Dobrushin-type conditions for weak mixing and ergodicity in the non-homogeneous setting~\cite{D56,mukhamedov2015ergodic, V24}.

Following the classical framework of Dobrushin \cite{D56} and subsequent works~\cite{Bau76,Sa15,mukhamedov2021approximations}, the asymptotic behaviour of inhomogeneous Markov chains is typically analysed via Dobrushin-type ergodicity coefficients. In this perspective, “mixing’’ is primarily a \emph{quantitative} concept—capturing convergence rates through mixing times and contraction bounds in total variation—rather than an autonomous structural notion distinct from ergodicity.

We call the chain \emph{weakly ergodic} (weakly mixing) if
\begin{equation}\label{eq:weak_mixing_inhom}
    \lim_{k\to\infty}\sup_{\mu,\nu \in \mathcal{P}_d}
    \big\|\mu\Pi^{(0,k)} - \nu\Pi^{(0,k)}\big\|_{\mathrm{tv}} = 0
\end{equation}
This means that, asymptotically, the system “forgets’’ its initial distribution.  An alternative viewpoint is obtained by embedding the process into the infinite product space $\mathcal{P}_d^{\mathbb{N}}$ and formulating mixing as asymptotic decay of correlations of distant coordinates. In this setting, classical mixing conditions~\cite{BDX06,L22, Hu96} assert that translating one event far along the time axis renders it asymptotically independent of any fixed event, and this pathwise notion of decorrelation has been extended to the quantum context for quantum Markov chains~\cite{SSB23} and open quantum systems~\cite{MixOQS23}.
\medskip

A stronger notion is \emph{ergodicity} (strong ergodicity) for the inhomogeneous chain. We say the chain is ergodic if there exists a probability vector $\pi \in \mathcal{P}_d$ such that
\begin{equation}\label{eq:strong_ergodic_TV}
    \lim_{k\to\infty}\sup_{\mu \in \mathcal{P}_d}
    \big\|\mu\Pi^{(0,k)} - \pi\big\|_{\mathrm{tv}} = 0
\end{equation}
Thus, ergodicity requires that the chain not only forget its initial state (as in~\eqref{eq:weak_mixing_inhom}), but also that all trajectories converge to a well-defined limiting law $\pi$. Weak mixing and ergodicity are related by the basic inequality
\begin{equation}\label{eq:triangle_ineq}
    \big\|\mu\Pi^{(0,k)} - \nu\Pi^{(0,k)}\big\|_{\mathrm{tv}}
    \leq 2\,\sup_{\mu'\in\mathcal{P}_d}\big\|\mu'\Pi^{(0,k)} - \pi\big\|_{\mathrm{tv}}
\end{equation}
valid whenever~\eqref{eq:strong_ergodic_TV} holds. The following Example illustrates the image in a particularly simple two-state case.
\begin{example}\label{exp1}\cite{Breamaud20} Consider a binary-state system ($d=2$) with time-dependent transition matrices that alternate between two forms:
\begin{align*}
\Pi_{2n} &= \begin{pmatrix}2^{-n} & 1-2^{-n} \\2^{-n} & 1-2^{-n}\end{pmatrix}  \\
\Pi_{2n+1} & = \begin{pmatrix}1-(2n+1)^{-1} & (2n+1)^{-1} \\1-(2n+1)^{-1} & (2n+1)^{-1}\end{pmatrix}\end{align*} for $n \geq 0$
This system exhibits weak ergodicity since $\|\mu\Pi^{(m,k)} - \nu\Pi^{(m,k)}\| \to 0$ as $k \to \infty$ for any initial distributions $\mu,\nu \in \mathcal{P}_2$. However, it fails to be strongly ergodic because the sequences $\mu\Pi_{2n}$ and $\mu\Pi_{2n+1}$  converges to two different limiting distributions $\pi_{even} = \left(\begin{array}{cc} 0 & 1\end{array}\right)$ and $\pi_{odd} = \left(\begin{array}{cc} 1 & 0 \end{array}\right)$  respectively.
\end{example}

A complementary, trajectory-level perspective is provided by ergodic (Cesàro) averages, which connect pathwise laws of large numbers with the operator-level notions of weak and uniform ergodicity discussed above. In the nonstationary setting, the convergence of Cesàro averages of the inhomogeneous transition operators, as well as of empirical means of observables, has been analyzed in several works including~\cite{Bau76,V22,V25}, thereby linking qualitative ergodicity to quantitative mixing and LLN-type behaviour along sample paths.

A natural first impulse in quantum information theory is to carry over the classical notions of weak mixing and ergodicity for inhomogeneous Markov chains to inhomogeneous quantum processes by simply replacing stochastic matrices with quantum channels. Yet even in the homogeneous classical setting, ergodicity and mixing already hinge in a subtle way on irreducibility and aperiodicity: uniqueness of the invariant distribution by itself does not guarantee convergence of \(\Pi^n\), as periodic chains such as the Ehrenfest model clearly demonstrate \cite{Ehrenfest1907}. In the homogeneous quantum case, by contrast, uniqueness (under standard primitivity-type assumptions) of an invariant density matrix typically \emph{does} entail mixing \cite{bau2013, W12}, so the structural link between invariant states, ergodicity, and mixing is genuinely different from the classical situation. In the inhomogeneous quantum regime this divergence becomes even more pronounced, and it is therefore not adequate to define quantum weak ergodicity and weak mixing by a naive transcription of classical Markov-chain notions. Instead, in this work we adopt an explicitly ergodic-theoretic perspective on quantum dynamics, formulated in terms of positive maps on ordered Banach spaces and their noncommutative extensions~\cite{mukhamedov2015ergodic,mukhamedov2021approximations}, and we will systematically use ergodic averages of iterated channel compositions to capture quantum ergodicity, while quantum mixing will be expressed directly through the composition of channels over multiple time steps, in direct coherence with the homogeneous theory of ergodic and mixing quantum channels.

\subsection{Forward and backward quantum Dynamics}
 Inhomogeneous quantum dynamics are described by sequences of quantum channels $\{\Phi_n\}_{n\geq0}$ acting on the set of density operators $\mathfrak{S}(\mathcal{H})$, where the quantum channel $\Phi_n$ can be written in terms of Kraus operators
 \begin{equation}\label{phinKrauss}
 \Phi_n(\rho) = \sum_{i_n}K_{n;i_n}\rho K_{n;i_n}^{\dagger}
 \end{equation}

The key distinction with the homogeneous dynamics  lies in the non-commutativity of channel compositions $\Phi_{n}\circ\Phi_{n+1}$ and $\Phi_{n+1}\circ\Phi_{n}$ generally differ, making temporal ordering crucial. This distinction leads to two fundamentally distinct types of dynamics. The backward evolution, denoted by $\Phi^{(b)}_{1,n}$, and given by:

\begin{equation}\label{eq-bwrd}
\Phi^{(b)}_{1,n} := \Phi_1 \circ \cdots \circ  \Phi_{n-1} \circ \Phi_n
\end{equation}

This composition describes the conventional time-ordered quantum process, where each operation $\Phi_k$ is applied in sequence from initial time $1$ to final time $n$. This evolution is governed by the equation $\Phi^{(b)}_{1,n+1} = \Phi^{(b)}_{1,n}\circ\Phi_{n+1}$. The Kraus decomposition associated with the backward dynamics is given by
\[
\Phi^{(b)}_{1,n}(\rho) =  \sum_{  i_1,\cdots, i_n}  K_{1;i_1}\cdots K_{n;i_n}\rho K_{n;i_n}^{\dagger}\cdots K_{1;i_1}^{\dagger}
\]
In contrast, the forward evolution $\Phi^{(f)}_{1,n}$ implements the reverse temporal ordering:

\begin{equation}\label{eq_forward}
\Phi^{(f)}_{1,n} := \Phi_n \circ \Phi_{n-1} \circ \cdots \circ \Phi_1
\end{equation}
 and the evolution is governed by the forward equation $\Phi^{(f)}_{1, n+1} = \Phi_{n+1}\circ \Phi_{1,n}^{(f)}$.
 The associated Kraus decomposition is given by
 \[
\Phi^{(f)}_{1,n}(\rho) =  \sum_{i_1,\cdots, i_n} K_{n;i_n}\cdots K_{1;i_1}  \rho K_{1;i_1}^{\dagger}\cdots K_{n;i_n}^{\dagger}\]

\begin{remark}
A fundamental difference between homogeneous and inhomogeneous quantum dynamics arises from the inherently non-commutative structure of quantum operations. In the time-dependent setting, quantum evolutions described by a sequence of channels $\{\Phi_n\}_{n \geq 1}$ exhibit qualitative behaviors that diverge significantly from those governed by a fixed channel $\Phi$. While homogeneous systems typically converge to fixed points determined by the invariant states of $\Phi$, inhomogeneous evolutions reflect the full temporal structure of the sequence, resulting in asymptotic behaviors that are sensitive to the entire history of the dynamics~\cite{aravinda2024ergodic, bau2013,Carbonna20}.

This distinction is particularly pronounced in the context of quantum control and open quantum systems, where external fluctuations and environmental interactions naturally give rise to time-dependent dynamics \cite{NP12}. Importantly, in the inhomogeneous case, standard fixed-point arguments become inapplicable due to the non-commutativity of the channels involved. As a result, classical characterizations of ergodicity, which rely on invariance under a single channel, no longer hold~\cite{nelson2024ergodic}. Instead, new analytical tools are required to capture the emergent long-term behavior of such non-stationary quantum processes.
\end{remark}
\begin{definition}[Ergodicity Hierarchy with Averages]\label{definition:ergodicity_hierarchy}
Let $\{\Phi_n\}_{n\in\mathbb{N}}$ be a sequence of quantum channels acting on the space of density operators $\mathfrak{S}(\mathcal{H})$, where $t \in \{f, b\}$ indicates the temporal direction of evolution. The \emph{ergodic average} is defined as:
\begin{equation}\label{equation:ergodic_average}
\overline{\Phi}^{(t)}_{1,n} := \frac{1}{n+1}\sum_{k=0}^{n}\Phi^{(t)}_{1,k}
\end{equation}

The system exhibits the following hierarchical ergodic properties:
\begin{enumerate}
    \item \textbf{Ergodicity}:
    \begin{equation}\label{equation:def_ergodicity}
    \exists \rho_{\infty}^{(t)} \in \mathfrak{S}(\mathcal{H});\ \forall \rho \in \mathfrak{S}(\mathcal{H}):\ \lim_{n\to\infty} \left\|\overline{\Phi}^{(t)}_{1,n}(\rho) - \rho_{\infty}^{(t)}\right\|_{1} = 0
    \end{equation}

    \item \textbf{Uniform Ergodicity}:
    \begin{equation}\label{equation:def_uniform_ergodicity}
    \exists \rho_{\infty}^{(t)} \in \mathfrak{S}(\mathcal{H}):\ \lim_{n\to\infty} \sup_{\rho \in \mathfrak{S}(\mathcal{H})} \left\|\overline{\Phi}^{(t)}_{1,n}(\rho) - \rho_{\infty}^{(t)}\right\|_{1} = 0
    \end{equation}

    \item \textbf{Mixing}:
    \begin{equation}\label{equation:def_mixing}
    \exists \rho_{\infty}^{(t)} \in \mathfrak{S}(\mathcal{H});\ \forall \rho \in \mathfrak{S}(\mathcal{H}):\ \lim_{n\to\infty} \Big\| \Phi^{(t)}_{1,n}(\rho) - \rho_{\infty}^{(t)}\Big\|_{1} = 0
    \end{equation}

    \item \textbf{Uniform Mixing}:
    \begin{equation}\label{equation:def_uniform_mixing}
    \exists \rho_{\infty}^{(t)} \in \mathfrak{S}(\mathcal{H}):\ \lim_{n \to \infty} \sup_{\rho \in \mathfrak{S}(\mathcal{H})} \Big\| \Phi^{(t)}_{1,n}(\rho) - \rho_{\infty}^{(t)} \Big\|_{1} = 0
    \end{equation}

    \item \textbf{Exponential Mixing}:
    \begin{equation}\label{equation:def_exponential_mixing}
    \exists \rho_{\infty}^{(t)} \in \mathfrak{S}(\mathcal{H}),\ \exists \mu\in(0,1),\ \exists C>0;\ \forall \rho \in \mathfrak{S}(\mathcal{H}):\ \Big\|\Phi^{(t)}_{1,n}(\rho) - \rho_{\infty}^{(t)}\Big\|_{1} \leq C \mu^n
    \end{equation}
\end{enumerate}
\end{definition}

\begin{remark}\label{remark:mixing_limit_channel}
    When the sequence $\{\Phi_{1,n}^{(t)}\}_{n \in \mathbb{N}}$ satisfies an ergodic or mixing property in the sense of Definition~\ref{definition:ergodicity_hierarchy}, the limiting density operator $\rho_{\infty}^{(t)}$ defines a limiting quantum channel $\Phi_{\infty}^{(t)}: \mathcal{B}(\mathcal{H}) \to \mathcal{B}(\mathcal{H})$ given by:
    \begin{equation}\label{eq:limiting_channel}
        \Phi_{\infty}^{(t)}(M) = \Tr(M) \rho_{\infty}^{(t)} \qquad \forall M \in \mathcal{B}(\mathcal{H})
    \end{equation}
    This channel is a rank-one, trace-preserving linear map that projects any operator onto the asymptotic state $\rho_{\infty}^{(t)}$.    Two critical differences arise compared to the homogeneous case, where dynamics are governed by a single repeated channel $\Phi$:
\begin{itemize}
    \item \textbf{Temporal direction dependence}: In the inhomogeneous setting, the limiting state $\rho_{\infty}^{(t)}$ is no longer universal, but depends explicitly on the evolution direction $t \in \{\text{forward}, \text{backward}\}$ associated with the sequence $\{\Phi_{1,n}^{(t)}\}_{n \in \mathbb{N}}$.
    \item \textbf{Sensitivity to perturbations}: In contrast to the homogeneous situation—where $\rho_{\infty}^{(t)}$ is a fixed point of the single channel $\Phi$ and is therefore stable under further iterations (see, e.g., \cite{singh2024ergodic,W12})—the asymptotic state in the time-inhomogeneous regime is typically fragile: even mild perturbations of, or reorderings within, the sequence may obstruct convergence or substantially modify the limit $\rho_{\infty}^{(t)}$.
\end{itemize}
This highlights a fundamental distinction between time-dependent and time-independent quantum dynamical systems.

\end{remark}
 \begin{theorem}[Implications within the Ergodicity Hierarchy with Averages]\label{theorem:ergodicity_hierarchy_implications}
Let $\{\Phi_n\}_{n\in\mathbb{N}}$ be a sequence of quantum channels acting on the space of density operators $\mathfrak{S}(\mathcal{H})$, and let $\overline{\Phi}^{(t)}_{1,n}$ denote the ergodic average as defined in Definition~\ref{definition:ergodicity_hierarchy}. Then, the following implications between the ergodic properties hold:
\begin{center}
$\text{Exponential Mixing} \Rightarrow \text{Uniform Mixing} \Leftrightarrow \text{Mixing} \Rightarrow \text{Uniform Ergodicity} \Leftrightarrow \text{Ergodicity}$.
\end{center}

Moreover, none of the reverse implications hold in general.
\end{theorem}

\begin{figure}[bh]
\centering
\begin{tikzpicture}[
  node distance=1.8cm and 1.2cm,
  base/.style={draw, rounded corners, minimum width=3.6cm, minimum height=1cm, align=center, font=\sffamily},
  imply/.style={double equal sign distance, -{Implies}, thick, draw=blue!60!black},
  equiv/.style={double, double distance=1.5pt, thick, draw=green!60!black, -{Implies}},
  horizequiv/.style={double, double distance=1.5pt, thick, draw=blue!60!black, {Implies[]}-{Implies[]}},
]

\node[base, fill=blue!5] (expMix) {Exponential\\Mixing};
\node[base, fill=blue!5] (unifMix) [right=of expMix] {Uniform\\Mixing};
\node[base, fill=blue!5] (mix) [right=of unifMix] {Mixing};

\node[base, fill=red!5] (unifErgo) [below=of unifMix] {Uniform\\Ergodicity};
\node[base, fill=red!5] (ergo) [right=of unifErgo] {Ergodicity};

\draw[imply] (expMix) -- (unifMix);
\draw[horizequiv] (unifMix) -- (mix); 
\draw[horizequiv] (unifErgo) -- (ergo); 

\draw[equiv] (unifMix.south) -- (unifErgo.north);
\draw[equiv] (mix.south) -- (ergo.north);

\end{tikzpicture}
\caption{Mixing and Ergodicity Hierarchy for Inhomogeneous Quantum Processes in Finite Dimension}
\end{figure}
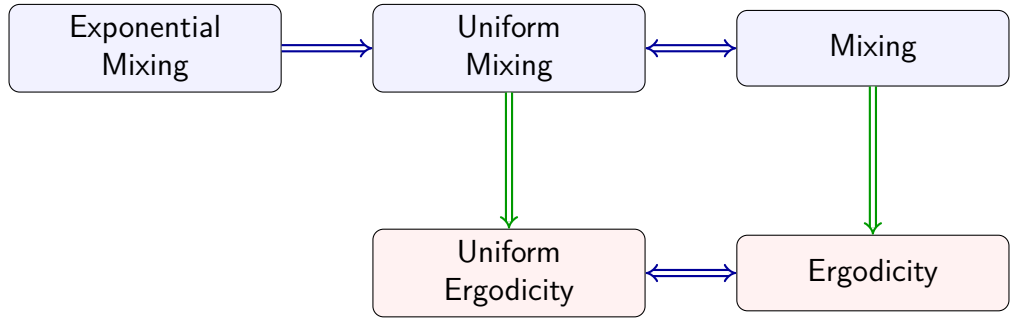
\begin{proof}
The only non-trivial implications are $\text{Mixing} \Rightarrow \text{Uniform Mixing}$ and $\text{Ergodic} \Rightarrow \text{Uniform Ergodic}$.

We first establish $\text{Mixing} \Rightarrow \text{Uniform Mixing}$ by contradiction. Suppose that
\[
\sup_{\rho \in \mathfrak{S}(\mathcal{H})} \Big\|\Phi^{(t)}_{1,n}(\rho) - \rho^{(t)}_{\infty}\Big\|_{1} \nrightarrow 0
\]
Then there exist $\varepsilon_0 > 0$, a subsequence $(\Phi^{(t)}_{1,n_k})_{k}$, and a sequence of states $(\rho_{n_k})_{k} \subset \mathfrak{S}(\mathcal{H})$ such that
\[
\bigl\|\Phi^{(t)}_{1,n_k}(\rho_{n_k}) - \rho^{(t)}_{\infty}\bigr\|_{1} \ge \varepsilon_0 \qquad \forall k
\]
Since $\dim(\mathcal{H}) < \infty$, the state space $\mathfrak{S}(\mathcal{H})$ is compact, and hence, up to extracting a subsequence, we may assume that $\rho_{n_k} \to \rho$ for some density operator $\rho \in \mathfrak{S}(\mathcal{H})$.

For the subsequence $\{\rho_{n_k}\}$, we decompose
\[
\bigl\|\Phi^{(t)}_{1,n_k}(\rho_{n_k}) - \rho^{(t)}_{\infty}\bigr\|_{1}
\leq
\underbrace{\bigl\|\Phi^{(t)}_{1,n_k}(\rho_{n_k}) - \Phi^{(t)}_{1,n_k}(\rho)\bigr\|_{1}}_{(A)}
+
\underbrace{\bigl\|\Phi^{(t)}_{1,n_k}(\rho) - \rho^{(t)}_{\infty}\bigr\|_{1}}_{(B)}
\]

\begin{itemize}
    \item \emph{Term (B)}: By the mixing assumption for the fixed state $\rho$,
    \[
        \bigl\|\Phi^{(t)}_{1,n_k}(\rho) - \rho^{(t)}_{\infty}\bigr\|_{1} \longrightarrow 0
        \quad \text{as } k \to \infty
    \]

    \item \emph{Term (A)}: Since the channels  are contractive with respect to the trace norm, one gets
    \[
        \bigl\|\Phi^{(t)}_{1,n_k}(\rho_{n_k}) - \Phi^{(t)}_{1,n_k}(\rho)\bigr\|_{1}
        \leq \|\rho_{n_k} - \rho\|_{1}
        \longrightarrow 0
        \quad \text{as } k \to \infty
    \]
    since $\rho_{n_k} \to \rho$ in trace norm.
\end{itemize}

Thus both (A) and (B) tend to zero, and we obtain
\[
\varepsilon_0 \le \bigl\|\Phi^{(t)}_{1,n_k}(\rho_{n_k}) - \rho^{(t)}_{\infty}\bigr\|_{1} \longrightarrow 0
\]
which is a contradiction.

The implication $\text{Ergodic} \Rightarrow \text{Uniformly Ergodic}$ is proved in an entirely analogous way, replacing $\Phi^{(t)}_{1,n}$ with the ergodic averages $\overline{\Phi}^{(t)}_{1,n}$. Counterexamples showing that the reverse implications fail to hold, even in the homogeneous case, can be found in~\cite{SB25,W12,singh2024ergodic,bau2013}.
\end{proof}
In the homogeneous setting, mixing of quantum Markov dynamics generated by a single channel is typically accompanied by an exponential rate of convergence under mild spectral or contractivity assumptions. In the inhomogeneous case, however, the situation can be markedly different: even for very simple families of commuting channels with a unique invariant state, mixing may occur at a strictly non-exponential rate. The following construction provides an explicit and fully tractable example of this phenomenon, thereby showing that the implication “uniformly mixing $\Rightarrow$ exponentially mixing’’ in Theorem~\ref{theorem:ergodicity_hierarchy_implications} cannot be reversed in general.
The following example makes this observation precise.

\begin{example}[Non-exponential mixing in inhomogeneous quantum dynamics]
\label{ex:polynomial_mixing}
Let $\mathcal{H}$ be a finite-dimensional Hilbert space with $d = \dim \mathcal{H}$, and let $\mathfrak{S}(\mathcal{H})$ denote the set of density operators on $\mathcal{H}$. Consider the family $\{\Phi_n\}_{n \in \mathbb{N}}$ of depolarizing channels defined by
\[
\Phi_n(\rho) = \Bigl(1 - \frac{1}{n+1}\Bigr)\rho + \frac{1}{n+1}\frac{\mathbb{I}}{d}, \qquad \rho \in \mathfrak{S}(\mathcal{H}), \ n \geq 1
\]
\medskip
The maximally mixed state $\frac{\mathbb{I}}{d}$ satisfies $\Phi_n(\frac{\mathbb{I}}{d}) = \frac{\mathbb{I}}{d}$ for all $n$, and is the unique invariant state of each $\Phi_n$. The family is commuting: $[\Phi_m,\Phi_n]=0$ for all $m,n$, as each $\Phi_n$ is an affine combination of the identity channel $\mathrm{id}$ and the completely depolarizing channel $\Omega(\rho)=\frac{\mathbb{I}}{d}$, which commute. Consequently, the forward and backward dynamics coincide.
\medskip
The $n$-step evolution from time $1$ to $n$ is $\Phi_{1,n}^{(t)} = \Phi_{1,n}^{(b)}=\Phi_{1,n}^{(f)}$, which by commutativity is independent of the order of composition. For any $\rho \in \mathfrak{S}(\mathcal{H})$:
\[
\Phi_{1,n}^{(t)}(\rho) = \Bigl(\prod_{k=1}^n \frac{k}{k+1}\Bigr)\rho + \Bigl(1 - \prod_{k=1}^n \frac{k}{k+1}\Bigr)\frac{\mathbb{I}}{d}
\]
The telescoping product simplifies to $\prod_{k=1}^n k/(k+1) = 1/(n+1)$, whence
\[
\Phi_{1,n}^{(t)}(\rho) = \frac{1}{n+1}\rho + \Bigl(1 - \frac{1}{n+1}\Bigr)\frac{\mathbb{I}}{d}, \qquad \forall \rho \in  \mathfrak{S}(\mathcal{H}), \ \forall n \in \mathbb{N}
\]
\medskip
In the trace norm $\|\cdot\|_1$,
\[
\bigl\|\Phi_{1,n}^{(t)}(\rho) - \tfrac{\mathbb{I}}{d}\bigr\|_1 = \frac{1}{n+1}\bigl\|\rho - \tfrac{\mathbb{I}}{d}\bigr\|_1
\]
Taking the supremum over all initial states,
\[
\Delta_n := \sup_{\rho \in  \mathfrak{S}} \bigl\|\Phi_{1,n}^{(t)}(\rho) - \tfrac{\mathbb{I}}{d}\bigr\|_1 = \frac{1}{n+1} \sup_{\rho \in  \mathfrak{S}(\mathcal{H})} \bigl\|\rho - \tfrac{\mathbb{I}}{d}\bigr\|_1 = \frac{2}{n+1}
\]
since $\sup_{\rho \in  \mathfrak{S}(\mathcal{H})} \|\rho - \frac{\mathbb{I}}{d}\|_1 = 2$, attained for any pure state orthogonal to a support projection of $\frac{\mathbb{I}}{d}$. Hence $\Delta_n = \mathcal{O}(1/n)$.
\medskip
No exponential estimate holds: there exist no constants $C>0$ and $\mu \in (0,1)$ such that $\Delta_n \leq C\mu^n$ for all $n$. Indeed, if such constants existed, then $2/(n+1) \leq C\mu^n$ for all $n$, contradicting $\lim_{n\to\infty} \mu^n (n+1) = 0$.
\medskip
The inhomogeneous family $\{\Phi_n\}_{n\geq 1}$ of commuting depolarizing channels converges uniformly in trace distance at rate $\mathcal{O}(1/n)$. This demonstrates that time-inhomogeneous quantum dynamics can exhibit genuinely polynomial convergence, in contrast to the homogeneous setting where spectral gaps guarantee exponential mixing. The construction is sharp: any rate $\mathcal{O}(n^{-\alpha})$ with $\alpha>0$ is achievable by choosing $p_n = 1/(n+1)^\alpha$.
\end{example}
\begin{remark}
In the finite-dimensional setting, the compactness of the state space $\mathfrak{S}(\mathcal{H})$ guarantees the equivalence between ergodicity and uniform ergodicity, as well as between mixing and uniform mixing. However, this equivalence generally breaks down in infinite-dimensional Hilbert spaces, where compactness is no longer present. A detailed proof of the equivalence in the finite-dimensional case is provided in the above proof of Theorem \ref{theorem:ergodicity_hierarchy_implications}. A thorough analysis of the distinctions arising in the infinite-dimensional scenario will be addressed in future work.
\end{remark}

\section{Weakly  Mixing Inhomogeneous Quantum processes}\label{sect_Mixcharacterization}
In this section, we introduce a hierarchy of ergodicity and mixing for inhomogeneous quantum processes, and prove Theorem \ref{theorem:characterizations_ergodicity_mixing}. We also present several related results, providing a rigorous framework for analyzing memory loss and convergence in time‑dependent quantum evolutions.
\begin{definition}[Weak Ergodicity, Weak Mixing, and Exponential Weak Mixing]\label{definition:weak_ergodicity_hierarchy}
Let $\{\Phi_n\}_{n \in \mathbb{N}}$ be a sequence of quantum channels on $\mathfrak{S}(\mathcal{H})$ and let $t \in \{f, b\}$ indicate forward ($f$) or backward ($b$) evolution.
\begin{enumerate}
    \item \textbf{Weak Ergodicity:} The sequence satisfies
    \begin{equation}\label{equation:weak_ergodicity}
        \forall \rho, \sigma \in \mathfrak{S}(\mathcal{H}): \; \lim_{n \to \infty}
        \bigl\| \overline{\Phi}^{(t)}_{1,n}(\rho) - \overline{\Phi}^{(t)}_{1,n}(\sigma) \bigr\|_{1} = 0
    \end{equation}
    \item \textbf{Uniform Weak Ergodicity:}
    \begin{equation}\label{equation:def_uniform_Weak_ergodicity}
        \lim_{n \to \infty} \sup_{\rho,\sigma \in \mathfrak{S}(\mathcal{H})}
        \bigl\| \overline{\Phi}^{(t)}_{1,n}(\rho) - \overline{\Phi}^{(t)}_{1,n}(\sigma) \bigr\|_{1} = 0
    \end{equation}
    \item \textbf{Weak Mixing:}
    \begin{equation}\label{equation:mixing}
        \forall \rho, \sigma \in \mathfrak{S}(\mathcal{H}): \; \lim_{n \to \infty}
        \bigl\| \Phi^{(t)}_{1,n}(\rho) - \Phi^{(t)}_{1,n}(\sigma) \bigr\|_{1} = 0
    \end{equation}
    \item \textbf{Uniform Weak Mixing:}
    \begin{equation}
        \lim_{n \to \infty} \sup_{\rho,\sigma \in \mathfrak{S}(\mathcal{H})}
        \bigl\| \Phi^{(t)}_{1,n}(\rho) - \Phi^{(t)}_{1,n}(\sigma) \bigr\|_{1} = 0
    \end{equation}

    \item \textbf{Exponential Weak Mixing:} There exist constants $C > 0$ and $\mu \in (0,1)$ such that
    \begin{equation}\label{equation:exponential_mixing}
        \sup_{\rho,\sigma \in \mathfrak{S}(\mathcal{H})}
        \bigl\| \Phi^{(t)}_{1,n}(\rho) - \Phi^{(t)}_{1,n}(\sigma) \bigr\|_{1} \leq C \mu^{\,n}
    \end{equation}
\end{enumerate}
\end{definition}
\begin{remark}[Time-uniform weak mixing on windows]
\label{remark:uniform_q_weak_mixing}
In addition to the weak ergodicity and weak mixing notions introduced in Definition~\ref{definition:weak_ergodicity_hierarchy} for compositions starting at a fixed initial time, it is natural to impose a stronger, \emph{time-uniform} (or \emph{uniform-in-window}) mixing property on the inhomogeneous dynamics. For a given direction \(t\in\{f,b\}\), we say that the evolution is time-uniformly weakly mixing on time windows if
\[
\forall m\ge 1:\quad
\lim_{n\to\infty}\sup_{\rho,\sigma\in\mathfrak{S}(\mathcal{H})}
\bigl\| {\Phi}^{(t)}_{m+1,n}(\rho)
      - {\Phi}^{(t)}_{m+1,n}(\sigma)\bigr\|_1 = 0
\]
While the standard weak mixing condition only requires
\[
\lim_{n\to\infty}\sup_{\rho,\sigma\in\mathfrak{S}(\mathcal{H})}
\bigl\| {\Phi}^{(t)}_{1,n}(\rho)
      - {\Phi}^{(t)}_{1,n}(\sigma)\bigr\|_1 = 0
\]
which quantifies loss of memory along a single time horizon emanating from the initial time \(0\), the time-uniform formulation demands that \emph{every} sufficiently long time interval \((m+1,n)\), regardless of its starting point \(m+1\), asymptotically acts as a contraction on the state space. In particular, the dynamics is not allowed to display good mixing only from the initial time and then recover memory on certain later blocks: the loss of memory must hold uniformly along the entire timeline.

The gap between these two levels of control is made explicit by a simple example that applies to both directions \(t\in\{b,f\}\). Let \(\mathcal{H}\) be a \(d\)-dimensional Hilbert space and define the completely depolarizing channel
\[
\Omega(\rho) :=  \,\frac{\mathbb{I}}{d},
\qquad \rho \in \mathfrak{S}(\mathcal{H})
\]
and consider the sequence of channels
\[
\Phi_1 = \Omega,\qquad \Phi_n = \mathrm{id}
\quad \text{for all } n \ge 1
\]
For either direction \(t\in\{f,b\}\), the composition \(\Phi^{(t)}_{1,n}\) is equal to \(\Omega\) for every \(n\ge 1\). Hence the ``from time \(0\)'' weak mixing condition holds in the strongest possible way: after the first step all initial states are mapped to the maximally mixed state and remain there. However, for any \(m\ge 1\) and any \(n\ge m+1\) we have
\[
\Phi^{(t)}_{m+1,n} = \mathrm{id}
\]
so on blocks that start strictly after the depolarizing step no contraction occurs at all. In particular,
\[
\sup_{\rho,\sigma\in\mathfrak{S}(\mathcal{H})}
\bigl\|\Phi^{(t)}_{n,n+q}(\rho)
      -\Phi^{(t)}_{n,n+q}(\sigma)\bigr\|_1
=
\begin{cases}
0, & n = 1,\\[0.3em]
2, & n \ge 2,
\end{cases}
\]
for every fixed \(q\in \mathbb{N}\) and both choices of \(t\). Thus the dynamics is weakly mixing in the usual sense (based on compositions starting at time \(0\)) for each direction, yet it fails the time-uniform weak mixing property on windows. This shows that the \((m+1,n)\)-based notion is strictly stronger than the standard weak mixing and uniform weak mixing definitions and genuinely enforces uniform mixing on \emph{all} time windows rather than only along a single trajectory.

The same idea naturally leads to strengthened \emph{time-uniform ergodic} and \emph{time-uniform mixing} requirements in which uniformity over time windows is imposed, respectively, on Cesàro averages of the compositions and on the compositions themselves. We do not develop the full hierarchy of such time-uniform ergodic and mixing properties here; instead, we simply note that the time-uniform weak mixing condition introduced above is the key asymptotic assumption used in Theorem~\ref{thm_main_inhom_MPS}, and that these refined ergodic and mixing notions will be systematically investigated in future work.
\end{remark}

\begin{remark}
The above definition captures three key modes of asymptotic behavior in the evolution of non-homogeneous quantum processes. Weak ergodicity reflects the idea that, when averaged over time, the influence of the initial state gradually disappears—a concept analogous to the classical ergodic theorem. Weak mixing strengthens this notion by requiring pointwise convergence of the evolved states, without averaging. Finally,  exponential weak mixing quantifies the rate of this convergence, demanding that it occurs at an exponential speed. These distinctions are crucial in understanding the degree to which a quantum system "forgets" its initial configuration and converge toward a universal long-term behavior, regardless of its starting point.
\end{remark}

\begin{theorem}[Implications within the Weak Ergodicity Hierarchy]\label{theorem:weak_ergodicity_hierarchy_implications}
Let $\{\Phi_n\}_{n\in\mathbb{N}}$ be a sequence of quantum channels on $\mathfrak{S}(\mathcal{H})$, and let $\overline{\Phi}^{(t)}_{1,n}$ denote the ergodic average as defined in Definition~\ref{definition:weak_ergodicity_hierarchy}. Then, the following implications hold for any evolution direction $t \in \{f,b\}$:
\begin{center}
$\text{Exponential Weak Mixing} \;\Rightarrow\; \text{Uniform Weak Mixing} \;\Leftrightarrow\; \text{Weak Mixing} \;\Rightarrow\; \text{Uniform Weak Ergodicity} \;\Leftrightarrow\; \text{Weak Ergodicity}$.
\end{center}
Moreover, none of the reverse implications holds in general.
\end{theorem}
\begin{proof}
We prove the implications in order. Throughout, \(\mathfrak{S}(\mathcal{H})\) is compact (finite‑dimensional) and the map \((\rho,\sigma)\mapsto \Big\|\Psi(\rho)-\Psi(\sigma)\Big\|_1\) is continuous for any channel \(\Psi\)

\noindent\textbf{(1) Exponential Weak Mixing \(\Rightarrow\) Uniform Weak Mixing.}
By definition, exponential weak mixing gives constants \(C>0\), \(\mu\in(0,1)\) with
\[
\sup_{\rho,\sigma}\Big\|\Phi^{(t)}_{1,n}(\rho)-\Phi^{(t)}_{1,n}(\sigma)\Big\|_1 \le C\mu^n
\]
Hence \(\lim_{n\to\infty}\sup_{\rho,\sigma}\Big\|\Phi^{(t)}_{1,n}(\rho)-\Phi^{(t)}_{1,n}(\sigma)\Big\|_1 = 0\), which is uniform weak mixing.

\noindent\textbf{(2) Uniform Weak Mixing \(\Leftrightarrow\) Weak Mixing.}

\noindent\textit{Uniform \(\Rightarrow\) weak.} This direction is immediate, since uniform convergence implies pointwise convergence.

\noindent\textit{Weak \(\Rightarrow\) uniform.}
Assume weak mixing holds, i.e.
\[
\forall \rho,\sigma\in\mathfrak{S}(\mathcal{H}):\;\lim_{n\to\infty}\Big\|\Phi^{(t)}_{1,n}(\rho)-\Phi^{(t)}_{1,n}(\sigma)\Big\|_1 = 0
\]
We treat the forward (\(t=f\)) and backward (\(t=b\)) cases separately.

\vspace{0.5em}
\noindent\textbf{Forward case (\(t=f\)).}
Define \(f_n(\rho,\sigma):=\Big\|\Phi_{1,n}^{(f)}(\rho)-\Phi_{1,n}^{(f)}(\sigma)\Big\|_1\).
Because each \(\Phi_k\) is continuous in the trace norm, each \(f_n\) is a continuous function on the compact space \(\mathfrak{S}(\mathcal{H})\times\mathfrak{S}(\mathcal{H})\).
Moreover, for any fixed \((\rho,\sigma)\),
\[
f_{n+1}(\rho,\sigma)=\Big\|\Phi_{n+1}\bigl(\Phi_{1,n}^{(f)}(\rho)\bigr)-\Phi_{n+1}\bigl(\Phi_{1,n}^{(f)}(\sigma)\bigr)\Big\|_1
\le\Big\|\Phi_{1,n}^{(f)}(\rho)-\Phi_{1,n}^{(f)}(\sigma)\Big\|_1=f_n(\rho,\sigma)
\]
where the inequality follows from the contractivity of \(\Phi_{n+1}\).
Thus \(\{f_n\}\) is a non‑increasing sequence of continuous functions converging pointwise to zero. By Dini’s theorem (monotone convergence on a compact set implies uniform convergence), we obtain
\[
\lim_{n\to\infty}\sup_{\rho,\sigma}f_n(\rho,\sigma)=0
\]
which is precisely uniform weak mixing.

\noindent\textbf{Backward case (\(t=b\)).}
For each fixed \(n\) define
\[F_n(\rho,\sigma)=\Big\|\Phi_{1,n}^{(b)}(\rho)-\Phi_{1,n}^{(b)}(\sigma)\Big\|_1\]
Define the \emph{diameter}
\[
d_n:=\sup_{\rho,\sigma\in\mathfrak{S}(\mathcal{H})}F_n(\rho,\sigma)
\]
Observe that
\[
\Phi_{1,n+1}^{(b)}=\Phi_{1,n}^{(b)}\circ\Phi_{n+1}
\]
Hence, for any \(\rho,\sigma\in\mathfrak{S}(\mathcal{H})\),
\[
\Big\|\Phi_{1,n+1}^{(b)}(\rho)-\Phi_{1,n+1}^{(b)}(\sigma)\Big\|_1
=\Big\|\Phi_{1,n}^{(b)}\bigl(\Phi_{n+1}(\rho)\bigr)-\Phi_{1,n}^{(b)}\bigl(\Phi_{n+1}(\sigma)\bigr)\Big\|_1
\le d_n
\]
because \(\Phi_{n+1}(\rho),\Phi_{n+1}(\sigma)\in\mathfrak{S}(\mathcal{H})\). Taking the supremum over \(\rho,\sigma\) yields \(d_{n+1}\le d_n\); thus \(\{d_n\}\) is a non‑increasing sequence of non‑negative numbers and therefore converges to some limit \(d_{\infty}\ge0\).

We now show that weak mixing forces \(d_{\infty}=0\). Fix \(\varepsilon>0\). By weak mixing, for each pair \((\rho,\sigma)\) there exists \(N(\rho,\sigma)\) such that
\[
n\ge N(\rho,\sigma)\;\Longrightarrow\; \Big\|\Phi_{1,n}^{(b)}(\rho)-\Phi_{1,n}^{(b)}(\sigma)\Big\|_1<\varepsilon
\]
For each \(n\), choose \((\rho_n,\sigma_n)\) attaining the diameter, i.e.
\[
d_n=\Big\|\Phi_{1,n}^{(b)}(\rho_n)-\Phi_{1,n}^{(b)}(\sigma_n)\Big\|_1 = F_n(\rho_n,\sigma_n)
\]
Since \(\mathfrak{S}(\mathcal{H})\times\mathfrak{S}(\mathcal{H})\) is compact, the sequence \(\{(\rho_n,\sigma_n)\}\) has a convergent subsequence \(\{(\rho_{n_k},\sigma_{n_k})\}\) with limit \((\rho_{\infty},\sigma_{\infty})\).
For this limit pair, weak mixing gives an \(N\) such that
\[
n\ge N\;\Longrightarrow\; F_n(\rho_\infty,\sigma_\infty) =\Big\|\Phi_{1,n}^{(b)}(\rho_{\infty})-\Phi_{1,n}^{(b)}(\sigma_{\infty})\Big\|_1<\varepsilon
\]
Now fix a \(k\) large enough so that \(n_k\ge N\), such that
\[
\Big\|\rho_{n_k} - \rho_{\infty}\Big\|_1 + \Big\|\sigma_{n_k} - \sigma_{\infty}\Big\|_1 < \varepsilon
\]
One has

\begin{eqnarray*}
d_{n_k} = F_{n_k}(\rho_{n_k},\sigma_{n_k})
&\le &  F_{n_k}(\rho_{n_k},\rho_{\infty}) +  F_{n_k}(\rho_\infty,\sigma_\infty) +  F_{n_k}(\sigma_\infty,\sigma_{n_k}),\\
 &\le &  \Big\|\rho_{n_k}-\rho_{\infty}\Big\|_1 + \varepsilon  + \Big\|\sigma_\infty-\sigma_{n_k}\Big\|_1
\end{eqnarray*}
Then
\[
d_{n_k}=F_{n_k}(\rho_{n_k},\sigma_{n_k})<F_{n_k}(\rho_{\infty},\sigma_{\infty})+\varepsilon<\varepsilon+\varepsilon=2\varepsilon
\]
Because \(\{d_n\}\) is non‑increasing, for all \(n\ge n_k\) we have \(d_n\le d_{n_k}<2\varepsilon\). Hence \(\lim_{n\to\infty}d_n=0\), i.e.
\[
\lim_{n\to\infty}\sup_{\rho,\sigma}\Big\|\Phi_{1,n}^{(b)}(\rho)-\Phi_{1,n}^{(b)}(\sigma)\Big\|_1=0
\]
which is uniform weak mixing.\\
\noindent\textbf{(3) Uniform Weak Ergodicity \(\Leftrightarrow\) Weak Ergodicity.}
The same compactness argument applied to the ergodic averages \(\overline{\Phi}^{(t)}_{1,n}\) yields the equivalence.

\noindent\textbf{(4) Weak Mixing \(\Rightarrow\) Uniform Weak Ergodicity.}
Assume weak mixing. Given \(\varepsilon>0\), choose \(N\) such that
\[
\Big\|\Phi^{(t)}_{1,n}(\rho)-\Phi^{(t)}_{1,n}(\sigma)\Big\|_1 < \varepsilon \quad \forall n\ge N,\; \forall\rho,\sigma
\]
For the ergodic average,
\[
\Big\|\overline{\Phi}^{(t)}_{1,n}(\rho)-\overline{\Phi}^{(t)}_{1,n}(\sigma)\Big\|_1
\le \frac{1}{n+1}\sum_{k=0}^{N-1} \Big\|\Phi^{(t)}_{1,k}(\rho)-\Phi^{(t)}_{1,k}(\sigma)\Big\|_1
+ \frac{1}{n+1}\sum_{k=N}^{n} \Big\|\Phi^{(t)}_{1,k}(\rho)-\Phi^{(t)}_{1,k}(\sigma)\Big\|_1
\]
The first sum is bounded by \(2N\) (diameter of \(\mathfrak{S}(\mathcal{H})\) is \(2\)). The second sum is bounded by \((n+1-N)\varepsilon\). Thus for large \(n\),
\[
\Big\|\overline{\Phi}^{(t)}_{1,n}(\rho)-\overline{\Phi}^{(t)}_{1,n}(\sigma)\Big\|_1
\le \frac{2N}{n+1} + \varepsilon
\]
Letting \(n\to\infty\) gives \(\limsup_n \Big\|\overline{\Phi}^{(t)}_{1,n}(\rho)-\overline{\Phi}^{(t)}_{1,n}(\sigma)\Big\|_1 \le \varepsilon\) for every \(\varepsilon>0\); hence the limit is zero. By compactness, the convergence is uniform, i.e., uniform weak ergodicity holds.
\end{proof}

\begin{remark}
The depolarizing sequence of Example~\ref{ex:polynomial_mixing} can be used analogously to show that \emph{uniform weak mixing does not imply exponential weak mixing}.
Indeed, the sequence
\[
\Phi_n(\rho) = \Bigl(1 - \frac{1}{n+1}\Bigr)\,\rho + \frac{1}{n+1}\,\frac{\mathbb{I}}{d}
\]
satisfies
\[
\sup_{\rho,\sigma\in\mathfrak{S}(\mathcal{H})}
\bigl\|\Phi^{(t)}_{1,n}(\rho) - \Phi^{(t)}_{1,n}(\sigma)\bigr\|_1
= \frac{2}{n+2}
\]
so it is uniformly weakly mixing.  However, the convergence is only polynomial, not exponential.  Hence it serves as a concrete counter‑example for the corresponding non‑reversible implication in Theorem \ref{theorem:weak_ergodicity_hierarchy_implications}.
\end{remark}
\begin{proof}[\textbf{\large Proof of Theorem \ref{theorem:characterizations_ergodicity_mixing}}]
   We start by showing that the direct implications hold both for the forward and backward dynamics. Let $\rho_{\infty}^{(t)}\in\mathfrak{S}(\mathcal{H})$  be the limiting density operators.

For every $\rho, \sigma\in\mathfrak{S}(\mathcal{H})$ one has

\begin{equation}\label{eq_triang}
\Big\|\Phi_{1,n}^{(t)}(\rho) - \Phi_{1,n}^{(t)}(\sigma) \Big\|_1 \le \Big\|\Phi_{1,n}^{(t)}(\rho) - \rho_{\infty}^{(t)} \Big\|_1 + \Big\|\Phi_{1,n}^{(t)}(\sigma) - \rho_{\infty}^{(t)} \Big\|_1
\end{equation}
It follows that if $\{\Phi_{1,n}^{(t)}\}$ is mixing then (\ref{equation:mixing}) is satisfied.

 Now for the exponential mixing  by majorizing  the right hand side of (\ref{eq_triang}) according to (\ref{equation:exponential_mixing}) one gets
\begin{equation*}
\Big\|\Phi_{1,n}^{(t)}(\rho) - \Phi_{1,n}^{(t)}(\sigma) \Big\|_1 \le 2C\mu^n
\end{equation*}
This finishes the direct implication for both backward and forward dynamics.

   $\Leftarrow$
    By definition of backward evolution, one has:
\begin{equation*}
\Phi_{1,n+1}^{(b)}(\mathfrak{S}(\mathcal{H}))
= \Phi_{1,n}^{(b)} \circ \Phi_{n+1} (\mathfrak{S}(\mathcal{H}))
\subseteq \Phi_{1,n}^{(b)}(\mathfrak{S}(\mathcal{H}))
\end{equation*}
Considering the nesting condition (\ref{eq_nest}) for the forward dynamics, it follows that the family of sets
\[
\mathcal{C}_n^{(t)} := \Phi_{1,n}^{(t)}(\mathfrak{S}(\mathcal{H}))
\]
forms a decreasing sequence of nonempty compact subsets of the compact space $\mathfrak{S}(\mathcal{H})$. Thus, by standard topological arguments, the intersection
\[
\bigcap_{n \in \mathbb{N}} \mathcal{C}_n^{(t)}
\]
is nonempty and compact. Assume now that the  condition \eqref{equation:mixing} holds. Then for any pair of states $\rho, \sigma \in \mathfrak{S}(\mathcal{H})$, one has
\[
\Big\|\Phi_{1,n}^{(t)}(\rho) - \Phi_{1,n}^{(t)}(\sigma)\Big\|_1  \xrightarrow{n \to \infty} 0
\]
Then thanks to the equivalence of weak mixing and uniform weak mixing in finite dimension from Theorem \ref{theorem:weak_ergodicity_hierarchy_implications}, the  sequence of diameters of the compact sets $\mathcal{C}_n^{(t)}$ goes to zero:
\[
\lim_{n \to \infty} \operatorname{diam}(\mathcal{C}_n^{(t)}) = 0
\]
As a result, the nested intersection contains a single point:
\[
\bigcap_{n \in \mathbb{N}} \mathcal{C}_n^{(t)} = \{\rho_{\infty}^{(t)}\}
\]
for some state $\rho_{\infty}^{(t)}\in \mathfrak{S}(\mathcal{H})$. This implies that for any initial state $\rho \in \mathfrak{S}(\mathcal{H})$, the sequence $\Phi_{1,n}^{(t)}(\rho)$ converges in norm to $\rho_{\infty}^{(t)}$, i.e.,
\[
\lim_{n \to \infty} \Phi_{1,n}^{(t)}(\rho) = \rho_{\infty}^{(t)}
\]
Thus, the process is indeed mixing with limiting state $\rho_{\infty}^{(t)}$.
\\

   The nesting condition implies that for every integers $m,n$ one has
     $$\mathcal{C}^{(t)}_{n+m} = \Phi^{(t)}_{1,n+m}(\mathfrak{S}(\mathcal{H}))\subseteq \Phi^{(t)}_{1,n}(\mathfrak{S}(\mathcal{H})) = \mathcal{C}^{(t)}_{n}, \quad \forall n,m\in\mathbb{N}$$

     Then for every $\rho\in\mathfrak{S}(\mathcal{H})$ there exists $\sigma\in\mathfrak{S}(\mathcal{H})$ such that $\Phi^{(t)}_{1,n+m}(\rho) = \Phi^{(t)}_{1,n}(\sigma)$. It follows from (\ref{equation:exponential_mixing}) that
   $$\Big\|\Phi_{1,n}^{(t)}(\rho) - \Phi_{1,n+m}^{(t)}(\rho)\Big\|_1 = \Big\|\Phi_{1,n}^{(t)}(\rho) - \Phi_{1,n}^{(t)}(\sigma)\Big\|_1\le C\mu^n
   $$
   then by taking $m\to \infty$ on the right hand side, one gets (\ref{equation:def_exponential_mixing}). This finishes the proof.
\end{proof}
\begin{example} \label{ex:quantum_mixing_oscillation}
This example extends Example~\ref{exp1} to the quantum setting, illustrating two key phenomena:
\begin{itemize}
    \item The equivalence between \emph{exponential mixing} and \emph{exponential weak mixing} in the \emph{backward dynamics} (as shown in  Theorem \ref{theorem:characterizations_ergodicity_mixing}).
    \item The failure of this equivalence in the \emph{forward direction} when the nesting condition (\ref{eq_nest}) is violated.
\end{itemize}

Consider a sequence of quantum channels $\{\Phi_n\}_{n\in\mathbb{N}}$ on $\mathcal{H} = \mathbb{C}^2$, defined by:
\begin{align*}
    \Phi_{2k}(\rho) &= |0\rangle\langle 0| \quad \text{(even steps)}\\
    \Phi_{2k+1}(\rho) &= |1\rangle\langle 1| \quad \text{(odd steps)}
\end{align*}
with Kraus operators $\{K_i^{(n)}\}$ satisfying $\sum_i K_i^{(n)\dagger} K_i^{(n)} = \mathbb{I}$ for all $n$.

\noindent \textbf{Forward Dynamics:}
The forward process $\Phi_{1,n}^{(f)} = \Phi_n \circ \cdots \circ \Phi_1$ exhibits:
\begin{equation} \label{eq:perfect_mixing}
    \Big\|\Phi_{1,n}^{(f)}(\rho) - \Phi_{1,n}^{(f)}(\sigma)\Big\|_1 = 0 \quad \forall \rho, \sigma \in \mathfrak{S}(\mathcal{H}), \; \forall n \geq 1.
\end{equation}
trivially satisfying both exponential weak mixing (\ref{equation:exponential_mixing}) and (yet) weak mixing (\ref{equation:mixing}). However, the sequence does not converge:
\begin{equation} \label{eq:oscillation}
    \Phi_{1,n}^{(f)}(\rho) = \begin{cases}
        |0\rangle\langle 0| & \text{if } n \text{ even}, \\
        |1\rangle\langle 1| & \text{if } n \text{ odd},
    \end{cases}
\end{equation}
due to the violation of the nesting condition (\ref{eq_nest}):
\begin{equation} \label{eq:nesting_failure}
    \Phi_{1,2k}^{(f)}(\mathfrak{S}(\mathcal{H})) = \{|0\rangle\langle 0|\} \not\subseteq \{|1\rangle\langle 1|\} = \Phi_{1,2k+1}^{(f)}(\mathfrak{S}(\mathcal{H}))
\end{equation}

\noindent \textbf{Backward Dynamics:}
In contrast, the backward process is stationary:
$$
\Phi_{1,n}^{(b)}(\rho) = \Phi_1 \circ \cdots \circ \Phi_n(\rho) = |1\rangle\langle 1|, \qquad \forall \rho\in\mathfrak{S}(\mathcal{H}).
$$
Then it is exponentially mixing, as the nesting condition holds trivially for the backward composition.

This demonstrates that while (\ref{eq:perfect_mixing}) ensures \emph{exponential weak mixing} in the forward direction, the lack of nesting prevents mixing. Thus, the nesting condition is essential for the equivalence of mixing notions in the forward case, whereas the backward dynamics naturally enforce it.
\end{example}

\begin{remark} \label{remark:backward_forward_asymmetry}
The structural analysis presented in Theorem~\ref{theorem:characterizations_ergodicity_mixing}   reveals an asymmetry between backward and forward dynamics in quantum processes. In the \emph{backward dynamics}, mixing and weak mixing coincide automatically due to an inherent structural nesting of the image spaces.  In the \emph{forward dynamics}, by contrast, the equivalence only holds under an additional nesting condition (\ref{eq_nest}).  Without this condition, weak mixing does not guarantee  mixing as illustrated in Example \ref{ex:quantum_mixing_oscillation}, revealing a fundamental difference between the two directions of time evolution.
\end{remark}

\begin{figure}[h]
 \begin{minipage}[c]{0.45\textwidth}
 \caption{Hierarchy of inclusion relationships among classes of homogeneous quantum processes  on $\mathfrak{S}(\mathcal{H})$:}\label{figerg}
\begin{itemize}
\item $\mathcal{EQP(\mathcal{H})}$: The set of ergodic quantum processes.
\item $\mathcal{MQP(\mathcal{H})}$: The set of mixing quantum processes.
\item $\mathcal{EMQP(\mathcal{H})}$: The set of exponentially mixing quantum processes.
\end{itemize}

\end{minipage}
    \hfill
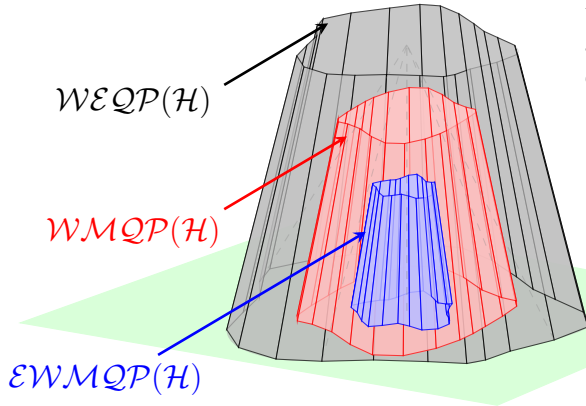
  \begin{minipage}[c]{0.5\textwidth}
\centering
\begin{tikzpicture}[scale=2.5]

\filldraw[green!20, opacity=0.4] (-1.3,-1.3) rectangle (1.3,1.3);
\node[black] at (0,-1.2) {\small Homogeneous Processes};

\draw[thick, fill=black!20]
    plot[smooth cycle, tension=0.6, domain=0:360, variable=\t]
    ({1.0*cos(\t)*(1.0+0.2*sin(3*\t))}, {1.0*sin(\t)*(1.0+0.1*cos(5*\t))});

\draw[thick, red, fill=red!20]
    plot[smooth cycle, tension=0.7, domain=0:360, variable=\t]
    ({0.85*cos(\t)*(0.8+0.15*sin(4*\t))}, {0.85*sin(\t)*(0.8+0.15*cos(2*\t))});

\draw[thick, blue, fill=blue!20]
    plot[smooth cycle, tension=0.8, domain=0:360, variable=\t]
    ({0.7*cos(\t)*(0.45+0.1*sin(5*\t))}, {0.7*sin(\t)*(0.45+0.1*cos(3*\t))});

\begin{scope}[shift={(0,-0.55,0)}]
    \draw[black, stealth-, line width=1pt] (0.4, 1.5) -- (1,1.7) ;
    \node[anchor=west, black] at (1, 1.7) {$\mathcal{EQP}(\mathcal{H})$};

    \draw[red, stealth-, line width=1pt] (0.3,1.1) -- (1.1,1.2);
    \node[anchor=west, red] at (1.1, 1.2) {$\mathcal{MQP}(\mathcal{H})$};

    \draw[blue, stealth-, line width=1pt]  (0.3,0.6) -- (1.1,0.5);
    \node[anchor=west, blue] at (1.1,0.5) {$\mathcal{EMQP}(\mathcal{H})$};
\end{scope}
\end{tikzpicture}
  \end{minipage}
  \\
  \begin{minipage}{0.45\textwidth}
\centering
\tdplotsetmaincoords{70}{300}
\begin{tikzpicture}[tdplot_main_coords, scale=2.5, xshift=-0.3cm]

\def\hBlack{1.5} 
\def\hRed{1.1}   
\def\hBlue{0.7}  

\filldraw[green!30, opacity=0.5, rotate around y=5]
    (-1.3,-1.3,0) -- (1.3,-1.3,0) -- (1.3,1.3,0) -- (-1.5,1.5,0) -- cycle;

\begin{scope}[canvas is xy plane at z=0]
    \path[black, fill=black!30, opacity=0.6]
        plot[smooth cycle, tension=0.6, domain=0:360, variable=\t]
        ({0.9*cos(\t)*(1.0+0.2*sin(3*\t))}, {0.9*sin(\t)*(1.0+0.1*cos(5*\t))});
\end{scope}
\begin{scope}[canvas is xy plane at z=\hBlack]
    \path[black, fill=black!15, opacity=0.8]
        plot[smooth cycle, tension=0.6, domain=0:360, variable=\t]
        ({0.55*cos(\t)*(1.0+0.2*sin(3*\t))}, {0.55*sin(\t)*(1.0+0.1*cos(5*\t))});
\end{scope}
\foreach \t in {0,15,...,345} {
    \draw[black, fill=black!25, opacity=0.7]
        ({0.9*cos(\t)*(1.0+0.2*sin(3*\t))}, {0.9*sin(\t)*(1.0+0.1*cos(5*\t))}, 0) --
        ({0.55*cos(\t)*(1.0+0.2*sin(3*\t))}, {0.55*sin(\t)*(1.0+0.1*cos(5*\t))}, \hBlack) --
        ({0.55*cos(\t+15)*(1.0+0.2*sin(3*(\t+15)))}, {0.55*sin(\t+15)*(1.0+0.1*cos(5*(\t+15)))}, \hBlack) --
        ({0.9*cos(\t+15)*(1.0+0.2*sin(3*(\t+15)))}, {0.9*sin(\t+15)*(1.0+0.1*cos(5*(\t+15)))}, 0) -- cycle;
}

\begin{scope}[canvas is xy plane at z=0]
    \path[red, fill=red!30, opacity=0.6]
        plot[smooth cycle, tension=0.7, domain=0:360, variable=\t]
        ({0.8*cos(\t)*(0.75+0.15*sin(4*\t))}, {0.8*sin(\t)*(0.75+0.15*cos(2*\t))});
\end{scope}
\begin{scope}[canvas is xy plane at z=\hRed]
    \path[red, fill=red!15, opacity=0.8]
        plot[smooth cycle, tension=0.7, domain=0:360, variable=\t]
        ({0.5*cos(\t)*(0.75+0.15*sin(4*\t))}, {0.5*sin(\t)*(0.75+0.15*cos(2*\t))});
\end{scope}
\foreach \t in {0,15,...,345} {
    \draw[red, fill=red!25, opacity=0.7]
        ({0.8*cos(\t)*(0.75+0.15*sin(4*\t))}, {0.8*sin(\t)*(0.75+0.15*cos(2*\t))}, 0) --
        ({0.5*cos(\t)*(0.75+0.15*sin(4*\t))}, {0.5*sin(\t)*(0.75+0.15*cos(2*\t))}, \hRed) --
        ({0.5*cos(\t+15)*(0.75+0.15*sin(4*(\t+15)))}, {0.5*sin(\t+15)*(0.75+0.15*cos(2*(\t+15)))}, \hRed) --
        ({0.8*cos(\t+15)*(0.75+0.15*sin(4*(\t+15)))}, {0.8*sin(\t+15)*(0.75+0.15*cos(2*(\t+15)))}, 0) -- cycle;
}

\begin{scope}[canvas is xy plane at z=0]
    \path[blue, fill=blue!30, opacity=0.6]
        plot[smooth cycle, tension=0.8, domain=0:360, variable=\t]
        ({0.65*cos(\t)*(0.4+0.1*sin(5*\t))}, {0.65*sin(\t)*(0.4+0.1*cos(3*\t))});
\end{scope}
\begin{scope}[canvas is xy plane at z=\hBlue]
    \path[blue, fill=blue!15, opacity=0.9]
        plot[smooth cycle, tension=0.8, domain=0:360, variable=\t]
        ({0.4*cos(\t)*(0.4+0.1*sin(5*\t))}, {0.4*sin(\t)*(0.4+0.1*cos(3*\t))});
\end{scope}
\foreach \t in {0,15,...,345} { \draw[blue, fill=blue!25, opacity=0.8]         ({0.65*cos(\t)*(0.4+0.1*sin(5*\t))}, {0.65*sin(\t)*(0.4+0.1*cos(3*\t))}, 0) --        ({0.4*cos(\t)*(0.4+0.1*sin(5*\t))}, {0.4*sin(\t)*(0.4+0.1*cos(3*\t))}, \hBlue) --        ({0.4*cos(\t+15)*(0.4+0.1*sin(5*(\t+15)))}, {0.4*sin(\t+15)*(0.4+0.1*cos(3*(\t+15)))}, \hBlue) --        ({0.65*cos(\t+15)*(0.4+0.1*sin(5*(\t+15)))}, {0.65*sin(\t+15)*(0.4+0.1*cos(3*(\t+15)))}, 0) -- cycle;}

\draw[gray, dashed, opacity=0.4] (0,0,0) -- (0,0,\hBlack);
\foreach \a in {0,90,180,270} {
    \draw[gray, dashed, opacity=0.3]
        ({0.9*cos(\a)*(1.0+0.2*sin(3*\a))}, {0.9*sin(\a)*(1.0+0.1*cos(5*\a))}, 0) --
        (0,0,\hBlack);
}

\begin{scope}[shift={(0,-0.55,0)}]
    \draw[black, stealth-, line width=1pt] (-0.4,0.8,1.7) -- (-1.5,0.8,1.7) ;
    \node[anchor=west, black] at (-3.4,0.8,2.2) {$\mathcal{WEQP}(\mathcal{H})$};

    \draw[red, stealth-, line width=1pt] (-0.2,0.8,1) -- (-1.5,0.8,1);
    \node[anchor=west, red] at (-3.5,0.8,1.5) {$\mathcal{WMQP}(\mathcal{H})$};

    \draw[blue, stealth-, line width=1pt]  (0,0.8,0.4) -- (-2.2,0.8,0.4);
    \node[anchor=west, blue] at (-3.9,0.8,0.8) {$\mathcal{EWMQP}(\mathcal{H})$};
\end{scope}
\end{tikzpicture}
\end{minipage} \hskip2mm
\begin{minipage}{0.45\textwidth}
 \caption{Hierarchy of the inclusion relationships among sets of weakly ergodic inhomogeneous quantum processes on $\mathfrak{S}(\mathcal{H})$.}\label{figweakerg} 
\begin{itemize}
\item $\mathcal{WEQP(\mathcal{H})}$: The set of weakly ergodic quantum processes.
\item $\mathcal{WMQP(\mathcal{H})}$: The set of weakly mixing quantum processes.
\item $\mathcal{EWMQP(\mathcal{H})}$: The set of exponentially weakly mixing quantum processes.
\end{itemize}
\end{minipage}
\end{figure}
\begin{remark}
The diagrams in Figs.~\ref{figerg} and~\ref{figweakerg} offer a schematic—but not literal—representation of the ergodic hierarchies for quantum processes. The geometric shapes are purely illustrative and do not reflect any intrinsic topology of the abstract sets themselves. Rather, they serve as visual metaphors: the 2D nested domains in Fig.~\ref{figerg} emphasize the strict, flat inclusions
$$\mathcal{EMQP}(\mathcal{H}) \subsetneq \mathcal{MQP}(\mathcal{H}) \subsetneq \mathcal{EQP}(\mathcal{H})
$$
for homogeneous processes, while the 3D conical structure in Fig.~\ref{figweakerg} captures the more intricate stratification
$$
\mathcal{EWMQP}(\mathcal{H}) \subsetneq \mathcal{WMQP}(\mathcal{H}) \subsetneq \mathcal{WEQP}(\mathcal{H})
$$
for inhomogeneous cases. The green plane in Fig.~\ref{figweakerg} symbolically represents how homogeneous processes (from Fig.~\ref{figerg}) arise as a "cross-section" of the weakly ergodic hierarchy, intersecting it at the level where time-dependent effects vanish. This contrast highlights how inhomogeneous processes generalize the homogeneous framework, with the added dimensionality reflecting their temporal complexity. No deeper geometric meaning is implied; the figures simply underscore that weak ergodicity accommodates a broader, more layered structure due to time dependence.
\end{remark}
\subsection{Trajectory-Ergodic Characterization of Weak Mixing}
\label{subsec:trajectory_ergodicity}
This subsection gives a sequential characterization of weak mixing and exponential weak mixing. Weak mixing is equivalent to the existence of a contracting sequence of density matrices \((Z_n)\) in the sense of Movassagh–Schenker \cite{movassagh2021theory,movassagh2022a}. Exponential weak mixing corresponds to a uniform exponential contraction rate along such a trajectory. This yields ergodicity criteria analogous to those in their framework, while include both forward and backward dynamics.
\begin{theorem}
\label{thm:mixing_ergodic_sequence}
Let $\{\Phi_n\}_n$ be a sequence of quantum channels and consider the associated sequences $\{\Phi_{1,n}^{(t)}\}_{n\in\mathbb{N}}$ with $\mathcal{C}_n^{(t)} = \Phi_{1,n}^{(t)}(\mathfrak{S}(\mathcal{H}))$ for $t \in \{f,b\}$. The following characterizations hold:

\begin{enumerate}
    \item \textbf{Weak Mixing Characterization}: The following are equivalent:
    \begin{enumerate}
        \item The sequence $\{\Phi_{1,n}^{(t)}\}_n$ is weakly mixing.
        \item For every sequence $\{Z_n\}_{n\in\mathbb{N}} \subset \mathfrak{S}(\mathcal{H})$ with $Z_n\in \mathcal{C}_n^{(t)}$, we have:
        \begin{equation}\label{eq:weak_mixing_uniform}
        \forall \rho\in\mathfrak{S}(\mathcal{H}), \quad \lim_{n\to\infty} \Big\|\Phi_{1,n}^{(t)}(\rho) - Z_n\Big\|_1 = 0
        \end{equation}
        \item There exists a sequence $\{Z_n\}_{n\in\mathbb{N}} \subset \mathfrak{S}(\mathcal{H})$ with $Z_n\in \mathcal{C}_n^{(t)}$ satisfying \eqref{eq:weak_mixing_uniform}.
    \end{enumerate}

    \item \textbf{ Exponential Weak Mixing Characterization}: The following are equivalent:
    \begin{enumerate}
        \item The sequence $\{\Phi_{1,n}^{(t)}\}_n$ is  exponentially weakly mixing.
        \item There exist constants $C > 0$ and $\mu\in(0,1)$ such that for every sequence $\{Z_n\}_{n\in\mathbb{N}} \subset \mathfrak{S}(\mathcal{H})$ with $Z_n\in \mathcal{C}_n^{(t)}$:
        \begin{equation}\label{eq:exp_mixing_uniform}
        \forall \rho\in\mathfrak{S}(\mathcal{H}), \quad \Big\|\Phi_{1,n}^{(t)}(\rho) - Z_n\Big\|_1 \leq C\mu^n \quad \forall n\in\mathbb{N}
        \end{equation}
        \item There exists a sequence $\{Z_n\}_{n\in\mathbb{N}} \subset \mathfrak{S}(\mathcal{H})$ with $Z_n\in \mathcal{C}_n^{(t)}$ and constants $C,\mu$ satisfying \eqref{eq:exp_mixing_uniform}.
    \end{enumerate}
\end{enumerate}
\end{theorem}

\begin{proof}
We prove both characterizations separately, leveraging the compactness of $\mathfrak{S}(\mathcal{H})$ and the properties of quantum channels.

\textbf{1. Weak Mixing Characterization}

\noindent \textbf{(a) $\Rightarrow$ (b):}
Assume $\{\Phi_{1,n}^{(t)}\}_n$ is weakly mixing.  From the equivalence of weak mixing and uniform weak mixing from  Theorem \ref{theorem:weak_ergodicity_hierarchy_implications} one gets the uniform convergence:
\[
\sup_{\rho, \sigma \in \mathfrak{S}(\mathcal{H})} \Big\|\Phi_{1,n}^{(t)}(\rho) - \Phi_{1,n}^{(t)}(\sigma)\Big\|_1 \to 0 \quad \text{as } n \to \infty
\]

For any sequence $\{Z_n\}_n$ with $Z_n \in \mathcal{C}_n^{(t)}$, there exists $\rho_n \in \mathfrak{S}(\mathcal{H})$ such that $Z_n = \Phi_{1,n}^{(t)}(\rho_n)$. For arbitrary $\rho \in \mathfrak{S}(\mathcal{H})$:
\[
\Big\|\Phi_{1,n}^{(t)}(\rho) - Z_n\Big\|_1 = \Big\|\Phi_{1,n}^{(t)}(\rho) - \Phi_{1,n}^{(t)}(\rho_n)\Big\|_1 \leq \sup_{\rho,\sigma} \Big\|\Phi_{1,n}^{(t)}(\rho) - \Phi_{1,n}^{(t)}(\sigma)\Big\|_1 \to 0
\]
This establishes (b).

\noindent \textbf{(b) $\Rightarrow$ (c):}
Immediate since (b) asserts the condition holds for all such sequences.

\noindent \textbf{(c) $\Rightarrow$ (a):}
Given a sequence $\{Z_n\}$ satisfying (c), for any $\rho, \sigma \in \mathfrak{S}(\mathcal{H})$:
\[
\Big\|\Phi_{1,n}^{(t)}(\rho) - \Phi_{1,n}^{(t)}(\sigma)\Big\|_1 \leq \Big\|\Phi_{1,n}^{(t)}(\rho) - Z_n\Big\|_1 + \Big\|\Phi_{1,n}^{(t)}(\sigma) - Z_n\Big\|_1 \to 0
\]
This  convergence implies the weak mixing condition (\ref{equation:mixing}).

\textbf{2. Exponential  Weak Mixing Characterization}

\noindent \textbf{(a) $\Rightarrow$ (b):}
Assume $\{\Phi_{1,n}^{(t)}\}_n$ is  exponentially weakly mixing. For any sequence $\{Z_n\}_n$ with $Z_n \in \mathcal{C}_n^{(t)}$, there exists $\rho_n \in \mathfrak{S}(\mathcal{H})$ such that $Z_n = \Phi_{1,n}^{(t)}(\rho_n)$. By (\ref{equation:exponential_mixing}), there exist $C > 0$, $\mu \in (0,1)$ such that:
\[
\Big\|\Phi_{1,n}^{(t)}(\rho) - Z_n\Big\|_1 = \Big\|\Phi_{1,n}^{(t)}(\rho) - \Phi_{1,n}^{(t)}(\rho_n)\Big\|_1 \leq C\mu^n \quad \forall \rho \in \mathfrak{S}(\mathcal{H})
\]
This holds for any sequence $\{Z_n\}$ since the bound is uniform.

\noindent \textbf{(b) $\Rightarrow$ (c):}
Trivial, as (b) is stronger.

\noindent \textbf{(c) $\Rightarrow$ (a):}
Given $\{Z_n\}$ satisfying (c), for any $\rho, \sigma \in \mathfrak{S}(\mathcal{H})$:
\[
\Big\|\Phi_{1,n}^{(t)}(\rho) - \Phi_{1,n}^{(t)}(\sigma)\Big\|_1 \leq \Big\|\Phi_{1,n}^{(t)}(\rho) - Z_n\Big\|_1 + \Big\|\Phi_{1,n}^{(t)}(\sigma) - Z_n\Big\|_1 \leq 2C\mu^n
\]
This establishes exponential mixing with constant $2C$ and rate $\mu$.
\end{proof}
 \begin{remark}
    Theorem~\ref{thm:mixing_ergodic_sequence} shows that weak mixing implies convergence to a recursively generated sequence $\{Z_n\}$ where
    $$
    Z_n = \Phi_n(Z_{n-1}) \in \Phi_{1,n}^{(t)}(\mathfrak{S}(\mathcal{H}))
    $$
     confirming results in \cite{movassagh2021theory, movassagh2022a}. This defines \emph{trajectory-ergodicity}: convergence occurs along a specific recursive path rather than to a static limit. The exponential case provides explicit convergence rates, while the non-exponential case retains the same recursive path structure without prescribed speed.
\end{remark}
\section{A Markov-Dobrushin Mixing Condition}\label{sect_MDmixing}
  This section outlines fundamental results on the behavior and properties of quantum channels, with a focus on their convergence rates via a quantum Markov-Dobrushin inequality.

For any operator $a \in \mathcal{B}(\mathcal{H})$, the decomposition $a = \frac{1}{2}(a + a^*) + i\frac{1}{2i}(a - a^*)$ yields its real part $\Re(a) = \frac{1}{2}(a + a^*)$ and imaginary part $\Im(a) = \frac{1}{2i}(a - a^*)$, both of which are self-adjoint. Any self-adjoint operator $b$ can be expressed as $b = b_+ - b_-$, where $b_+, b_-$ are positive operators with disjoint supports, and we define $\|b\|_1 = \Tr(b_+) + \Tr(b_-)$. Building on this, we introduce the total variation norm $\|\cdot\|_{TV}$ on $\mathcal{M}_d(\mathbb{C})$, that coincides  with   $\|\cdot\|_1$  on Hermitian operators. This norm is defined as:
\begin{equation}\label{df-q-tot-var-nrm}
    \|a\|_{TV} := \Tr\left(\Re(a)_+ + \Re(a)_-\right) + \Tr\left(\Im(a)_+ + \Im(a)_-\right)
\end{equation}
and has been shown to constitute a norm on $\mathcal{B}(\mathcal{H})$ as a real vector space \cite{AccLuSou22}.

\begin{lemma}\label{Lem_MD}
The Markov--Dobrushin infimum set $ \mathcal{I}^{\Tr}_{\mathrm{MD}}(\Phi)$
defined in \eqref{infMD} is nonempty.
\end{lemma}
\begin{proof}
Let \(X_{\Phi} := \{ \Phi(P) : P \in \mathcal{P}_1(\mathcal{H}) \} \subset \mathcal{B}(\mathcal{H})_+\). Since \(\dim \mathcal{H} = d \ge 1\), the set \(\mathcal{P}_1(\mathcal{H})\) is nonempty, hence \(X_{\Phi} \neq \varnothing\). Define \(L_+(X_{\Phi}) := \{ B \in \mathcal{B}(\mathcal{H})_+ : 0 \le B \le A \text{ for all } A \in X_{\Phi} \}\). Because \(0 \le 0 \le \Phi(P)\) for every rank-one projection \(P\), we have \(0 \in L_+(X_{\Phi})\); thus \(L_+(X_{\Phi})\) is nonempty.

To establish closedness of \(L_+(X_{\Phi})\) with respect to the trace norm \(\|\cdot\|_1\), fix an arbitrary \(A \in X_{\Phi}\) and consider \(C_A := \{ B \in \mathcal{B}(\mathcal{H})_+ : 0 \le B \le A \}\). Let \((B_n)_{n \in \mathbb{N}} \subset C_A\) converge in \(\|\cdot\|_1\) to some \(B \in \mathcal{B}(\mathcal{H})\). Since \(\mathcal{B}(\mathcal{H})\) is finite-dimensional, \(\|\cdot\|_1\) is equivalent to the operator norm \(\|\cdot\|\); hence \(B_n \to B\) in operator norm. Each \(B_n\) is self-adjoint, and the set of self-adjoint operators is closed in \(\|\cdot\|\), so \(B\) is self-adjoint. For any unit vector \(v \in \mathcal{H}\), we have \(\langle B_n v, v \rangle \ge 0\) for all \(n\) and \(|\langle (B_n - B)v, v \rangle| \le \|B_n - B\| \to 0\); thus \(\langle B v, v \rangle = \lim_n \langle B_n v, v \rangle \ge 0\), proving \(B \ge 0\). Applying the same argument to \(A - B_n \ge 0\) yields \(A - B \ge 0\), i.e., \(B \le A\). Therefore \(B \in C_A\), so each \(C_A\) is closed in \(\|\cdot\|_1\). Since \(L_+(X_{\Phi}) = \bigcap_{A \in X_{\Phi}} C_A\) and an arbitrary intersection of closed sets is closed, \(L_+(X_{\Phi})\) is closed in \(\|\cdot\|_1\).

For boundedness, choose any \(A_0 \in X_{\Phi}\). For every \(B \in L_+(X_{\Phi})\) we have \(0 \le B \le A_0\), so the eigenvalues \(\lambda_1(B), \dots, \lambda_d(B)\) satisfy \(0 \le \lambda_j(B) \le \lambda_{\max}(A_0)\). Consequently, \[\|B\|_1 = \operatorname{Tr}(B) = \sum_{j=1}^d \lambda_j(B) \le d \cdot \lambda_{\max}(A_0)\] showing that \(L_+(X_{\Phi})\) is bounded in \(\|\cdot\|_1\).

Since \((\mathcal{B}(\mathcal{H}), \|\cdot\|_1)\) is finite-dimensional, closedness and boundedness imply compactness of \(L_+(X_{\Phi})\). The map \(f : L_+(X_{\Phi}) \to \mathbb{R}\) defined by \(f(B) := \operatorname{Tr}(B)\) is linear and hence continuous with respect to \(\|\cdot\|_1\). By compactness, \(f\) attains its maximum on \(L_+(X_{\Phi})\); thus there exists \(\kappa_\Phi \in L_+(X_{\Phi})\) such that \(\operatorname{Tr}(\kappa_\Phi) = \max_{B \in L_+(X_{\Phi})} \operatorname{Tr}(B)\). By construction, \(\kappa_\Phi\) satisfies \(0 \le \kappa_\Phi \le \Phi(P)\) for all \(P \in \mathcal{P}_1(\mathcal{H})\) and has maximal trace among all such operators, which are precisely the defining properties of \(\mathcal{I}^{\Tr}_{\mathrm{MD}}(\Phi)\) given in \eqref{infMD}. Hence \(\kappa_\Phi \in \mathcal{I}^{\Tr}_{\mathrm{MD}}(\Phi)\), proving that the Markov--Dobrushin infimum set is nonempty.
\end{proof}
\begin{remark}
When the operator infimum \(\bigwedge_{P \in \mathcal{P}_1(\mathcal{H})} \Phi(P)\) exists in the operator order, the Markov--Dobrushin infimum set reduces to a singleton containing precisely this infimum. In that case, the unique element \(\kappa_\Phi = \bigwedge_{P \in \mathcal{P}_1(\mathcal{H})} \Phi(P)\) coincides with the Markov--Dobrushin constant introduced in \cite{AccLuSou22}, which characterizes the mixing behavior of homogeneous quantum dynamics \cite{SB25}.
\end{remark}

\begin{theorem}\label{thm_main}
    Let $\Phi: \mathcal{B}(\mathcal{H}) \to \mathcal{B}(\mathcal{H})$ be a quantum channel. Then, for any pair of states $\rho, \sigma \in \mathfrak{S}(\mathcal{H})$, the following inequality holds:
    \begin{equation}\label{q-MD-ineq}
        \|\Phi(\rho) - \Phi(\sigma)\|_{TV} \leq \left(1 - \Tr(\kappa_{\Phi})\right) \|\rho - \sigma\|_{TV}
    \end{equation}
    where   \(\kappa_{\Phi} \in  \mathcal{I}^{\Tr}_{\mathrm{MD}}(\Phi) \)
    is the quantum Markov-Dobrushin constant of $\Phi$.
\end{theorem}
\begin{proof}
  See \cite{AccLuSou22}.
\end{proof}

\noindent Recall that an \emph{accumulation point} (or \emph{limit point}) of a sequence \(\{x_n\}_{n=1}^\infty\) in a metric space \((X, d)\) is a point \(x \in X\) such that every neighborhood of \(x\) contains infinitely many terms of the sequence. Formally:
\[
\forall \varepsilon > 0, \quad B(x, \varepsilon) \cap \{x_n : n \in \mathbb{N}\} \text{ is infinite}
\]
where \(B(x, \varepsilon) = \{y \in X : d(x, y) < \varepsilon\}\).

Equivalently, \(x\) is an accumulation point if there exists a subsequence \(\{x_{\tau_n}\}_{n=1}^\infty \subset \{x_n\}_{n=1}^\infty\) such that:
\[
\lim_{n \to \infty} x_{\tau_n} = x
\]
\noindent Theorem \ref{thm-main_MD_mixing} establishes  weak mixing for non-homogeneous quantum dynamics under a compactness condition on the quantum Markov-Dobrushin constants. The key requirement is the existence of a non-zero accumulation point, which enables uniform control over state convergence in both forward and backward time directions.

\begin{proof}[\textbf{\large Proof of Theorem \ref{thm-main_MD_mixing}}]
We begin by recalling the fundamental contraction property of the Markov--Dobrushin coefficient. For any quantum channel $\Phi$, the associated coefficient $\kappa_{\Phi} \in \mathcal{I}^{\Tr}_{\mathrm{MD}}(\Phi)$ satisfies the inequality
\[
\sup_{\rho,\sigma \in \mathfrak{S}(\mathcal{H})} \big\|\Phi(\rho) - \Phi(\sigma)\big\|_{\mathrm{TV}} \le \big(1 - \Tr(\kappa_{\Phi})\big) \sup_{\rho,\sigma \in \mathfrak{S}(\mathcal{H})} \|\rho - \sigma\|_{\mathrm{TV}}
\]
Since the trace norm on the compact convex set $\mathfrak{S}(\mathcal{H})$ is bounded by $2$, this yields the pointwise contraction estimate
\[
\big\|\Phi(\rho) - \Phi(\sigma)\big\|_{\mathrm{TV}} \le \big(1 - \Tr(\kappa_{\Phi})\big) \|\rho - \sigma\|_{\mathrm{TV}} \qquad \forall \rho,\sigma \in \mathfrak{S}(\mathcal{H})
\]
Crucially, this contraction is independent of the order in which channels are composed. Consequently, for any finite composition of channels, whether taken in the forward direction $\Phi^{(f)}_{m+1,n} = \Phi_n \circ \cdots \circ \Phi_{m+1}$ or the backward direction $\Phi^{(b)}_{m+1,n} = \Phi_{m+1} \circ \cdots \circ \Phi_n$, we obtain by iterating the contraction inequality the uniform estimate
\begin{equation}\label{eq:prod-estimate}
\big\|\Phi^{(t)}_{m+1,n}(\rho) - \Phi^{(t)}_{m+1,n}(\sigma)\big\|_{\mathrm{TV}}
\le \left( \prod_{j=m+1}^{n} \big(1 - \Tr(\kappa_{\Phi_j})\big) \right) \|\rho - \sigma\|_{\mathrm{TV}}
\end{equation}
valid for both $t \in \{f,b\}$ and any indices $m < n$.

Now, the hypothesis that $\{\kappa_{\Phi_n}\}$ admits a non-zero accumulation point ensures the existence of a number $r > 0$ and an infinite set of indices $\{n_k\}_{k \in \mathbb{N}}$ such that $\Tr(\kappa_{\Phi_{n_k}}) \ge r$ for every $k$. For each integer $n$, we define
\[
N(n) := \#\{\,k \in \mathbb{N} \mid n_k \le n\,\}
\]
i.e., the count of ``good'' channels—those whose Markov--Dobrushin coefficient has trace at least $r$—appearing among the first $n$ positions. The function $n \mapsto N(n)$ is non-decreasing, integer-valued, and tends to infinity as $n \to \infty$ because there are infinitely many such indices. Moreover, $N(n)$ is unbounded and strictly increasing along the subsequence of good channels; by passing to a suitable subsequence if necessary, we may assume without loss of generality that $N(n)$ itself is strictly increasing, thereby satisfying the requirement of an increasing unbounded sequence.

Turning to the product appearing in \eqref{eq:prod-estimate}, we observe that each factor satisfies $1 - \Tr(\kappa_{\Phi_j}) \le 1$, while for each good index $j = n_k$ we have the improved bound $1 - \Tr(\kappa_{\Phi_j}) \le 1 - r$. Hence, for any $m < n$, the product over the interval $(m, n]$ can be bounded above by the product over only the good indices within that interval, yielding
\[
\prod_{j=m+1}^{n} \big(1 - \Tr(\kappa_{\Phi_j})\big) \le (1 - r)^{\,N(n) - N(m)}
\]
Setting $\mu := 1 - r$, we note that $0 < \mu < 1$ by construction. Substituting this estimate into \eqref{eq:prod-estimate} and employing the trivial bound $\|\rho - \sigma\|_{\mathrm{TV}} \le 2$ for any two density operators, we obtain
\[
\big\|\Phi^{(t)}_{m+1,n}(\rho) - \Phi^{(t)}_{m+1,n}(\sigma)\big\|_{\mathrm{TV}} \le 2 \mu^{\,N(n) - N(m)} \qquad \forall \rho,\sigma \in \mathfrak{S}(\mathcal{H})
\]
Since the right-hand side is independent of the particular choice of $\rho$ and $\sigma$, we may take the supremum over all pairs of states, which precisely yields the desired inequality \eqref{eq_exp_conv}. This establishes the existence of $\mu \in (0,1)$ and the increasing unbounded sequence $\{N(n)\}_{n \in \mathbb{N}}$ satisfying the exponential contraction bound for both the forward and backward evolutions.

To see that this implies weak mixing, fix an arbitrary $m \in \mathbb{N}$ and consider the limit as $n \to \infty$. Because $\{N(n)\}$ is unbounded, we have $N(n) \to \infty$ as $n \to \infty$, and consequently $\mu^{N(n) - N(m)} \to 0$. Hence,
\[
\lim_{n \to \infty} \sup_{\rho,\sigma \in \mathfrak{S}(\mathcal{H})} \Big\|\Phi^{(t)}_{m+1,n}(\rho) - \Phi^{(t)}_{m+1,n}(\sigma)\Big\|_{\mathrm{TV}} = 0
\]
This is precisely the definition of weak mixing for the family $\{\Phi^{(t)}_{m+1,n}\}_{n > m}$ as stated in the theorem: for each fixed $m$, the distance between states evolved from times $m+1$ to $n$ tends to zero uniformly as $n \to \infty$, regardless of the initial states $\rho$ and $\sigma$. Moreover, the inequality \eqref{eq_exp_conv} provides a quantitative rate: the distance decays at least exponentially fast in $N(n) - N(m)$, i.e., in the number of good channels encountered between times $m$ and $n$. This completes the proof.
\end{proof}

 \begin{remark}
Theorem~\ref{thm-main_MD_mixing} establishes a unified framework for understanding mixing phenomena in inhomogeneous quantum processes, significantly generalizing important previous works on ergodic quantum processes. Our results reveal that what was previously characterized as ergodicity in the forward dynamics corresponds precisely to what we identify as exponential mixing. The convergence rate in (\ref{eq_exp_conv}), expressed as
\[
e^{-N(n)\ln(1/\mu)}
\]
encompasses a spectrum of dynamical behaviors - from exponential to polynomial convergence - depending on the density $N(n)$ of effective mixing channels. For instance, when $N(n) \sim \ln(n)$, we obtain polynomial convergence of order $n^{-\alpha}$ with $\alpha = \ln(1/\mu)$. This raises the Markov-Dobrushin mixing condition recovers a strictly larger class of mixing inhomogeneous processes that the class of exponentially mixing quantum processes  through the positivity assumptions  suggested in \cite{movassagh2022a}.
\end{remark}

\begin{corollary}[Exponential Mixing in Periodic Quantum  Processes]\label{cor:periodic}
    Let $\{\Phi_n\}_{n \in \mathbb{N}}$ be a sequence of quantum channels with period $p \in \mathbb{N}^*$, i.e., $\Phi_{n+p} = \Phi_n$ for all $n \in \mathbb{N}$. Suppose there exists some index $1 \leq j_* \leq p$ such that the Markov-Dobrushin constant of $\Phi_{j_k}$ satisfies $\kappa_{\Phi_{j_k}} > 0$.
    Then, the  quantum process generated by $\{\Phi_n\}$ exhibits  exponential weak mixing in both forward ($t = f$) and backward ($t = b$) time directions. Specifically, there exists a constant $C = C(p, \kappa_{\Phi_{j_k}}) > 0$ such that for all $n \in \mathbb{N}$ and any pair of quantum states $\rho, \sigma \in \mathfrak{S}(\mathcal{H})$,
    \begin{equation}\label{eq:periodic_mixing}
        \left\| \Phi^{(t)}_{1,n}(\rho) - \Phi^{(t)}_{1,n}(\sigma) \right\|_{\mathrm{TV}} \leq C \mu^n
    \end{equation}
    where the contraction rate $\mu$ is given explicitly by
    $  \mu = \left(1 - \kappa_{\Phi_{j_k}}\right)^{\frac{1}{p}}$
\end{corollary}
\begin{proof}
The $p$-periodicity of the channel sequence allows us to analyze the mixing properties through cycle decomposition. For any $n \in \mathbb{N}$, let $n = qp + s$ where $q = \lfloor n/p \rfloor$ counts complete periods and $0 \leq s < p$ counts remaining operations.

Let $j_{k} \in [1,p]$  such that  $\kappa_{\Phi_{j_k}}>0$ then  each period contains at least a channel whose contraction coefficient is equal to $\kappa_{\Phi_{j_k}}$. The cumulative contraction satisfies:

\[
\prod_{j=1}^n (1 - \Tr(\kappa_{\Phi_j}))
= \underbrace{\left(\prod_{j=1}^p (1 - \Tr(\kappa_{\Phi_j}))\right)^q}_{\text{periodic contraction}}
\cdot \underbrace{\prod_{\ell=1}^s (1 - \Tr(\kappa_{\Phi_\ell}))}_{\text{residual terms}}
\leq (1 - \Tr(\kappa_{\Phi_{j_k}}))^q
\]
 One has
\[
(1- \Tr(\kappa_{\Phi_{j_k}}))^q  = (1- \Tr(\kappa_{\Phi_{j_k}}))^{\frac{n-s}{p}} = \frac{1}{2}C \mu^{n}
\]
where  $C= 2(1- \Tr(\kappa_{\Phi_{j_k}}))^{\frac{-s}{p}}$ and $\mu = (1- \Tr(\kappa_{\Phi_{j_k}}))^{\frac{1}{p}}$. It follows  that:
 $$
 \| \Phi_{1,n}^{(t)}(\rho) - \Phi_{1,n}^{(t)}(\sigma)\|_{TV} \le \prod_{j=1}^n (1 - \Tr(\kappa_{\Phi_j})) \| \rho - \sigma\|_{TV}\le  C\mu^{n}
  $$
This leads to the exponential weak mixing and finishes the proof.
\end{proof}

\section{Mixing Properties for Inhomogeneous Matrix Product States} \label{sect_AppMPS}
\label{sec:mps_applications}
Over the past few years, MPS have increasingly been studied through an algebraic perspective of quantum processes \cite{SB25,SA26}, providing a natural setting for investigating ergodic and mixing behaviors—especially beyond the translation-invariant case~\cite{ perez2007matrix}. This dynamical viewpoint interprets the local tensors of an MPS as quantum channels, enabling a refined analysis of their long-term properties using tools from quantum ergodic theory. Notably, the works~\cite{albert2019asymptotics, SB25} demonstrate that such a formulation is particularly effective in describing exponential mixing and in capturing the thermodynamic limit behavior of MPS.

Let $\mathcal{H}$ be a finite-dimensional Hilbert space of dimension $d$, called the \emph{bond space}, with orthonormal basis $\{e_k\}_{k=1}^d$. Let $\mathcal{K}$ be a finite-dimensional Hilbert space of dimension $d_{\mathcal{K}}$, called the \emph{physical space}, with orthonormal basis $\{|i\rangle\}_{i=1}^{d_{\mathcal{K}}}$. The algebra of bounded linear operators on $\mathcal{H}$ is denoted $\mathcal{B}(\mathcal{H})$, and analogously $\mathcal{B}(\mathcal{K})$ for the physical space.

For each site $n\in\mathbb{N}$, we are given a collection of operators $\{K_i^{[n]}\}_{i=1}^{d_{\mathcal{K}}}\subset\mathcal{B}(\mathcal{H})$, called \emph{site-dependent tensors}, which satisfy the left-canonical (or gauge) condition
\begin{equation}\label{eq:left_canonical}
\sum_{i=1}^{d_{\mathcal{K}}} (K_i^{[n]})^\dagger K_i^{[n]} = \mathbb{I}, \qquad \forall n\in\mathbb{N}
\end{equation}
These tensors define a family of quantum channels $\Phi_n:\mathcal{B}(\mathcal{H})\to\mathcal{B}(\mathcal{H})$ in the Schrödinger picture by
\begin{equation*}\label{eq:channel}
\Phi_n(\rho) := \sum_{i=1}^{d_{\mathcal{K}}} K_i^{[n]} \rho (K_i^{[n]})^\dagger, \qquad \rho\in\mathcal{B}(\mathcal{H})
\end{equation*}
The condition \eqref{eq:left_canonical} ensures that each $\Phi_n$ is trace-preserving and unital, hence a quantum channel. The backward composition of channels from site $m+1$ to $n$ is defined as
\begin{equation*}\label{eq:backward_composition}
\Phi_{m+1,n}^{(b)} := \Phi_{m+1} \circ \Phi_{m+2} \circ \cdots \circ \Phi_n
\end{equation*}
with the convention that $\Phi_{n+1,n}^{(b)} = \mathrm{id}_{\mathcal{B}(\mathcal{H})}$. Explicitly,
\[
\Phi_{m+1,n}^{(b)}(M) = \sum_{k_{m+1},\ldots,k_n}
K_{k_{m+1}}^{[m+1]} \cdots K_{k_n}^{[n]} \,
M \,
(K_{k_n}^{[n]})^\dagger \cdots (K_{k_{m+1}}^{[m+1]})^\dagger
\]
For a superoperator $\Psi : \mathcal{B}(\mathcal{H}) \to \mathcal{B}(\mathcal{H})$, we define its \emph{superoperator trace} (or trace over the operator space) by
\begin{equation}\label{eq:superoperator_trace}
\operatorname{Tr}_{\text{sup}}(\Psi) := \sum_{k,l=1}^d e_k^* \, \Psi(e_k e_l^*) \, e_l
\end{equation}
 This quantity is independent of the chosen basis.\\
 For a finite chain of length $n$, the local observable algebra is given by the tensor product
\[
\mathfrak{A}_{[1,n]} := \bigotimes_{k=1}^n \mathcal{B}(\mathcal{K}) \cong \mathcal{B}(\mathcal{C}^{\otimes n})
\]
These algebras embed into larger systems via the inclusion maps
\[
\iota_n : \mathfrak{A}_{[1,n]} \hookrightarrow \mathfrak{A}_{[1,n+1]}, \qquad X \mapsto X \otimes \mathbb{I}_{\mathcal{K}}
\]
and the inductive limit yields the quasi-local algebra \cite{BR}
\[
\mathfrak{A}_{\text{loc}} := \varinjlim \mathfrak{A}_{[1,n]} = \bigcup_{n\ge 1} \mathfrak{A}_{[1,n]}
\]
whose norm completion is the C*-algebra of the semi-infinite chain:
\[
\mathfrak{A}_{\mathbb{N}} := \overline{\bigotimes_{\mathbb{N}} \mathcal{B}(\mathcal{K})}^{\|\cdot\|}
\]
The inhomogeneous matrix product state (MPS) on $n$ sites is the pure state vector
\begin{equation}\label{eq:MPS_vector}
|\Psi_n\rangle := \sum_{i_1,\ldots,i_n=1}^{d_{\mathcal{K}}} \operatorname{Tr}\bigl(K_{i_1}^{[1]} K_{i_2}^{[2]} \cdots K_{i_n}^{[n]}\bigr) \; |i_1\cdots i_n\rangle \in \mathcal{C}^{\otimes n}
\end{equation}
The corresponding state functional $\varphi_n : \mathfrak{A}_{[1,n]} \to \mathbb{C}$ is defined by
\begin{equation}\label{eq:state_functional}
\varphi_n(X) :=\frac{\langle \Psi_n | X | \Psi_n \rangle}{\langle\Psi_n|\Psi_n\rangle}
\end{equation}
where the normalization of the MPS is given by
\begin{equation}\label{eq:norm_calculation}
\langle\Psi_n|\Psi_n\rangle = \|\Psi_n\|^2 = \sum_{i_1,\ldots,i_n} \bigl|\operatorname{Tr}(K_{i_1}^{[1]}\cdots K_{i_n}^{[n]})\bigr|^2 = \operatorname{Tr}\bigl(\Phi_{1,n}^{(b)})
\end{equation}
where the last equality follows from the left-canonical condition \eqref{eq:left_canonical} and the trace-preserving property of each $\Phi_k$.

For any $m\in\mathbb{N}$ and any local observable $X\in\mathfrak{A}_{[1,m]}$, we associate to $X$ a linear map
\[
\widehat{X} : \mathcal{B}(\mathcal{H}) \to \mathcal{B}(\mathcal{H})
\]
evaluated on $ M\in\mathcal{B}(\mathcal{H})$ by
\begin{equation}\label{eq:observable_map}
\widehat{X}(M)
  := \sum_{\substack{i_1,\ldots,i_m\\ j_1,\ldots,j_m}}
     \langle i_1,\ldots,i_m|X|j_1,\ldots,j_m\rangle\;
     K_{j_1}^{[1]} \cdots K_{j_m}^{[m]}\, M \, (K_{i_1}^{[1]})^\dagger \cdots (K_{i_m}^{[m]})^\dagger
\end{equation}
The placement of the indices in this definition is deliberate. The Kraus operators carrying the bra indices $j_1,\ldots,j_m$ are placed to the \emph{left} of $M$, while those carrying the ket indices $i_1,\ldots,i_m$ are placed to the \emph{right} of $M$, conjugated. Thus the ordering in \eqref{eq:observable_map} is reversed with respect to the natural order of the matrix elements
\(\langle i_1,\ldots,i_m|X|j_1,\ldots,j_m\rangle\).

This choice is not cosmetic: it is precisely what makes $\widehat{X}$ compatible with the channel structure of the MPS. If one takes $X=\mathbb{I}$, the map $\widehat{\mathbb{I}}$ coincides with the \emph{backward} composition of channels
\[
\Phi_{1,m}^{(b)} := \Phi_1\circ\cdots\circ\Phi_m,\qquad
\Phi_{1,m}^{(b)}(M)
  = \sum_{i_1,\ldots,i_m}
      K_{i_1}^{[1]}\cdots K_{i_m}^{[m]} M (K_{i_m}^{[m]})^\dagger\cdots(K_{i_1}^{[1]})^\dagger
\]
which is the standard MPS transfer operator built from the tensors $K_i^{[n]}$.

In the general case of a nontrivial observable $X$, the flipped index placement in \eqref{eq:observable_map} is dictated by the contraction pattern of the MPS when computing local expectation values. As will be shown in the proof of Theorem~\ref{thm_main_inhom_MPS}, for any $n\ge m$ the expectation $\varphi_n(X)$ can be rewritten in the form
\[
\varphi_n(X)
  = \frac{\operatorname{Tr}_{\text{sup}}\bigl(\widehat{X}\circ \Phi_{m+1,n}^{(b)}\bigr)}{\operatorname{Tr}_{\text{sup}}\bigl( \Phi_{1,n}^{(b)}\bigr)}
\]
where
\[
\Phi_{m+1,n}^{(b)} := \Phi_{m+1}\circ\cdots\circ\Phi_n
\]
is the \emph{backward} composition of the channels associated with the MPS tensors. The point is that the MPS contraction naturally produces this backward transfer operator, and the particular ordering in \eqref{eq:observable_map} ensures that $\widehat{X}$ couples to it in a clean, channel-theoretic way.

This is exactly the regime in which   Theorem \ref{theorem:characterizations_ergodicity_mixing} and Theorem \ref{thm-main_MD_mixing} will be applied: under the Markov--Dobrushin assumptions, the backward composition $\{\Phi_{m+1,n}^{(b)}\}_{n\ge m+1}$ is time-uniformly mixing and converges to a completely depolarizing channel. 
With these notations established, we now present the proof of Theorem \ref{thm_main_inhom_MPS}.

 \begin{proof}
For any $X\in\mathfrak{A}_{[0,m]}$ and integer $n\ge m$, one computes the expectation value
\begin{align*}
  \langle \Psi_n | X\otimes \mathbb{I}_{[m+1,n]} | \Psi_n \rangle
  &=
  \sum_{\substack{i_1,\ldots,i_n\\ j_1,\ldots,j_n}}
    \overline{\Tr\!\bigl(K^{[1]}_{i_1} K^{[2]}_{i_2} \cdots K^{[n]}_{i_n}\bigr)}
    \,\Tr\!\bigl(K^{[1]}_{j_1} K^{[2]}_{j_2} \cdots K^{[n]}_{j_n}\bigr)\,\\
   &\qquad\qquad \times \langle i_1,\ldots,i_n|X\otimes \mathbb{I}_{[m+1,n]}|j_1,\ldots,j_n\rangle
\end{align*}

Using the Kronecker deltas to perform the sums over $i_{m+1},\ldots,i_n$ and $j_{m+1},\ldots,j_n$ and regrouping the Kraus operators, one obtains
\[
 \langle \Psi_n | X\otimes \mathbb{I}_{[m+1,n]} | \Psi_n \rangle
 =\sum_{\substack{i_1,\dots,i_n\\ j_1,\dots,j_n}}
\overline{\operatorname{Tr}\bigl(K_{i_1}^{[0]}\cdots K_{i_n}^{[n]}\bigr)}\,
\operatorname{Tr}\bigl(K_{j_1}^{[0]}\cdots K_{j_n}^{[n]}\bigr)\,
\langle i_1,\dots,i_m|X|j_1,\dots,j_m\rangle\,
\prod_{t=m+1}^n\delta_{i_tj_t}
\]
Summing over $i_{m+1},\dots,i_n$ using the Kronecker deltas yields
\begin{align*}
\langle \Psi_n | X\otimes \mathbb{I}_{[m+1,n]} | \Psi_n\rangle &= \sum_{\substack{i_1,\dots,i_m\\ j_1,\dots,j_m}}
\langle i_1,\dots,i_m|X|j_1,\dots,j_m\rangle\,\\
&\qquad\qquad\times\sum_{\ell_{m+1},\dots,\ell_n}
\overline{\operatorname{Tr}\bigl(A\,K_{\ell_{m+1}}^{[m+1]}\cdots K_{\ell_n}^{[n]}\bigr)}\,
\operatorname{Tr}\bigl(B\,K_{\ell_{m+1}}^{[m+1]}\cdots K_{\ell_n}^{[n]}\bigr)
\end{align*}
where $A:=K_{i_1}^{[0]}\cdots K_{i_m}^{[m]}$ and $B:=K_{j_1}^{[0]}\cdots K_{j_m}^{[m]}$.

Using the basis $\{e_ke_l^*\}_{k,l=1}^d$ and the cyclic property of the trace, one verifies the identity
\[
\sum_{\ell_{m+1},\dots,\ell_n}
\operatorname{Tr}(B\,K)\;\overline{\operatorname{Tr}(A\,K)}
= \sum_{k,l=1}^d e_k^*\bigl((B\cdot A^\dagger)\circ\Phi_{m+1,n}^{(b)}(e_ke_l^*)\bigr)e_l
\]
where $(B\cdot A^\dagger)(M):=B M A^\dagger$. Summing over the indices $i_1,\dots,i_m,j_1,\dots,j_m$ with the coefficients $\langle i_1,\dots,i_m|X|j_1,\dots,j_m\rangle$ gives
\[
\langle \Psi_n | X\otimes \mathbb{I}_{[m+1,n]} | \Psi_n\rangle
\overset{\eqref{eq:observable_map}}{=}
\sum_{k,l=1}^d e_k^*\bigl(\widehat{X}\circ\Phi_{m+1,n}^{(b)}(e_ke_l^*)\bigr)e_l
\tag{$\star$}
\]

By Theorem \ref{thm-main_MD_mixing} and the hypothesis that the Markov--Dobrushin coefficients $\{\kappa_{\Phi_n}\}$ admit a non-zero accumulation point, the sequence $\{\Phi_{m+1,n}^{(b)}\}_{n\ge m}$ converges in the superoperator norm to the completely depolarizing channel:
\[
\lim_{n\to\infty}\Phi_{m+1,n}^{(b)} = \Omega_{\rho_{m+1}^{(b)}}
\]
where $\rho_{m+1}^{(b)}\in\mathfrak{S}(\mathcal{H})$ is a density matrix and $\Omega_{\rho}(M):=\operatorname{Tr}(M)\,\rho$ for all $M\in\mathcal{B}(\mathcal{H})$.
The convergence is uniform on bounded subsets, and each $\Phi_{m}^{(b)}$ is linear and continuous.

From the composition law $\Phi_{m,n}^{(b)} = \Phi_{m}\circ\Phi_{m+1,n}^{(b)}$, taking limits yields
\[
\Omega_{\rho_m^{(b)}} = \Phi_{m}\circ\Omega_{\rho_{m+1}^{(b)}}
\]
Evaluating both sides on an arbitrary $M\in\mathcal{B}(\mathcal{H})$ gives
\[
\operatorname{Tr}(M)\rho_m^{(b)} = \operatorname{Tr}(M)\,\Phi_{m}(\rho_{m+1}^{(b)})
\]
hence for all $m\in\mathbb{N}$,
\[
\rho_m^{(b)} = \Phi_{m}(\rho_{m+1}^{(b)})
\tag{$\diamond$}
\]
Since $\widehat{X}$ is a finite linear combination of maps of the form $B\cdot A^\dagger$, it is continuous with respect to the superoperator norm. Therefore, from \;$(\star)$\; and the convergence $\Phi_{m+1,n}^{(b)}\to\Omega_{\rho_{m+1}^{(b)}}$,
\[
\lim_{n\to\infty}\varphi_n(X)
= \sum_{k,l=1}^d e_k^*\bigl(\widehat{X}\circ\Omega_{\rho_{m+1}^{(b)}}(e_ke_l^*)\bigr)e_l
\]
Now compute the composition:
\[
\widehat{X}\circ\Omega_{\rho_{m+1}^{(b)}}(e_ke_l^*)
= \widehat{X}\bigl(\operatorname{Tr}(e_ke_l^*)\,\rho_{m+1}^{(b)}\bigr)
= \delta_{kl}\,\widehat{X}(\rho_{m+1}^{(b)})
\]
Taking into account \eqref{eq:state_functional} and \eqref{eq:norm_calculation}, one obtains
\[
\lim_{n\to\infty}\varphi_n(X)
= \frac{\sum_{k=1}^d e_k^*\,\widehat{X}(\rho_{m+1}^{(b)})\,e_k}
       {\sum_{k=1}^d e_k^*\,\Phi_{1,m}^{(b)}(\rho_{m+1}^{(b)})\,e_k}
\overset{\eqref{eq:superoperator_trace}}{=} \frac{\operatorname{Tr}\bigl(\widehat{X}(\rho_{m+1}^{(b)})\bigr)}
       {\operatorname{Tr}\bigl(\Phi_{1,m}^{(b)}(\rho_{m+1}^{(b)})\bigr)}
= \operatorname{Tr}\bigl(\widehat{X}(\rho_{m+1}^{(b)})\bigr)
\tag{$\heartsuit$}
\]

Inserting the definition of $\widehat{X}$ according to \eqref{eq:observable_map} into \;$(\heartsuit)$\; yields the explicit formula \eqref{eq_phiinfty_short}.

The right-hand side depends on $X$ only through its support $[0,m]$, and the sequence $\{\rho_m^{(b)}\}$ satisfies \;$(\diamond)$. Moreover, the convergence in \;$(\heartsuit)$\; holds for every $X\in\mathfrak{A}_{[0,m]}$ and every $m\in\mathbb{N}$, establishing weak-$\ast$ convergence of $\varphi_n$ to a state $\varphi_\infty$ on $\mathfrak{A}_{\mathbb{N}}$. The uniqueness of $\varphi_\infty$ follows from the uniqueness of the limit $\Omega_{\rho_{m+1}^{(b)}}$ guaranteed by the mixing property (Theorem \ref{thm-main_MD_mixing}). Thus the sequence $\{\varphi_n\}$ converges in the weak-$\ast$ topology to a unique state $\varphi_\infty$ on $\mathfrak{A}_{\mathbb{N}}$, and for every local observable $X\in\mathfrak{A}_{[0,m]}$ the expectation $\varphi_\infty(X)$ is given by \eqref{eq_phiinfty_short} with boundary states $\{\rho_m^{(b)}\}\subset\mathfrak{S}(\mathcal{H})$ satisfying $\rho_m^{(b)}=\Phi_{m+1}(\rho_{m+1}^{(b)})$ for all $m\in\mathbb{N}$.
\end{proof}

\begin{remark}
Theorem~\ref{thm_main_inhom_MPS} provides a clean, dynamical explanation of mixing   for inhomogeneous matrix product states purely in terms of the quantum Markov--Dobrushin condition on the associated channel sequence. It shows that the long-time behaviour of the backward dynamics $\{\Phi_{m+1,n}^{(b)}\}$ enforces convergence to a unique infinite-volume MPS state and determines its local expectations, without relying on translation invariance or classical ergodic arguments.

This viewpoint complements existing results on ergodic quantum processes and MPS (such as \cite{movassagh2022a}), and extends the homogeneous Markov–Dobrushin analysis of \cite{SB25} to genuinely non-translation-invariant tensor networks. It also fits naturally with recent asymptotic descriptions of quantum channels \cite{albert2019asymptotics}, and suggests a systematic framework for studying non-mixing or marginal regimes—where the Markov–Dobrushin coefficients cease to be uniformly bounded away from zero—potentially revealing phase transitions in disordered or non-uniform quantum spin systems.
\end{remark}

\section{Discussion and outlook}

This work develops a quantitative framework for analysing ergodicity and mixing in time-inhomogeneous quantum systems, described by a sequence of quantum channels \(\{\Phi_n\}_{n\in\mathbb{N}}\) acting on \(\mathcal{B}(\mathcal{H})\), with particular emphasis on their induced dynamics on the state space \(\mathfrak{S}(\mathcal{H})\). Building on the quantum Markov--Dobrushin (qMD) approach of~\cite{AccLuSou22}, we exploit contraction properties of the forward and backward compositions \(\Phi^{(t)}_{m,n}\), \(t\in\{b,f\}\), to derive verifiable conditions ensuring convergence of \(\Phi^{(t)}_{m,n}(\rho)\) as \(n\to\infty\) together with explicit quantitative bounds on the speed of convergence, beyond the purely exponential regime characteristic of homogeneous dynamics. The application to non-translation-invariant matrix product states (MPS) shows that these abstract qMD conditions can be implemented in concrete tensor-network models: under suitable assumptions on the sequence \(\{\Phi_n\}\), one obtains a unique infinite-volume state \(\varphi_\infty\) on the quasi-local algebra and an explicit channel-based representation of its local expectations \(\varphi_\infty(X)\) in terms of Kraus operators and boundary states.

From a conceptual standpoint, a key advantage of the qMD approach is that mixing is characterized in a way that is simultaneously time-inhomogeneous and \emph{time-uniform}. Concretely, the framework provides a family of Markov--Dobrushin coefficients \(\{\kappa_{\Phi_n}\}\subset\mathcal{B}(\mathcal{H})\) such that a single lower bound on \(\Tr(\kappa_{\Phi_n})\) controls the contraction of an appropriate quantum total variation distance, defined through the Markov--Dobrushin norm, over all sufficiently long time windows \([m,n]\), uniformly with respect to the starting time \(m\).  This yields a robust notion of loss of memory: the dependence on the initial state \(\rho\) is suppressed at a rate governed by the aggregated qMD coefficients, in a way that is intrinsically adapted to non-stationary dynamics and that extends classical weak-mixing ideas to a quantum, channel-based setting. In tensor-network language, the same mechanism explains how the auxiliary-space dynamics encoded by the transfer channels \(\Phi_n\) controls the thermodynamic limit of an inhomogeneous MPS, with the Markov--Dobrushin data determining decay of correlations and stability of local marginals.

\medskip

It is natural to compare this Markov--Dobrushin perspective with the more traditional viewpoint based directly on contraction coefficients of quantum channels with respect to a fixed metric or divergence. For a channel \(\Phi\), one can define the (worst-case) contraction coefficient for the trace distance by
\[
\eta_{\mathrm{Tr}}(\Phi)
:= \sup_{\rho\neq\sigma}
\frac{\|\Phi(\rho)-\Phi(\sigma)\|_1}{\|\rho-\sigma\|_1},
\]
which plays a central role in quantum information theory as a measure of information loss and distinguishability degradation under noise~\cite{George2025,Hirche2024, Delsol2025}. In our setting, the qMD construction produces a specific, computable contraction coefficient \(\eta_{\mathrm{MD}}(\Phi)\) such that
\[
\eta_{\mathrm{Tr}}(\Phi) \;\leq\; \eta_{\mathrm{MD}}(\Phi),
\]
and hence
\[
\|\Phi(\rho)-\Phi(\sigma)\|_1
\;\leq\; \eta_{\mathrm{MD}}(\Phi)\,\|\rho-\sigma\|_1
\qquad\text{for all }\rho,\sigma\in\mathfrak{S}(\mathcal{H}).
\]
Thus, the Markov--Dobrushin constant furnishes a specific channel contraction coefficient, expressed in terms of a single operator \(\kappa_\Phi\), whose trace controls the loss of distinguishability in the trace norm. At the same time, by construction \(\eta_{\mathrm{MD}}(\Phi)\) is not expected to be optimal: in general one can only guarantee that it bounds from above the true channel contraction coefficient \(\eta_{\mathrm{Tr}}(\Phi)\), which may be strictly smaller. The strength of the qMD coefficient lies elsewhere: it is often straightforward to compute or to bound directly from structural information on the Kraus operators or symmetries of \(\Phi_n\), and it behaves well under composition along time-inhomogeneous evolutions. The trade-off between computational accessibility and sharpness is therefore a structural design choice of the qMD method rather than a defect.

\medskip

The present analysis also has clear structural limitations. First, all results are formulated in a finite-dimensional setting, with \(\mathcal{H}\cong\mathbb{C}^d\) and channels \(\Phi_n : \mathcal{B}(\mathcal{H}) \to \mathcal{B}(\mathcal{H})\). In this regime, compactness of the state space \(\mathfrak{S}(\mathcal{H})\), equivalence of operator norms, and spectral properties of completely positive maps can be handled with standard finite-dimensional tools. Second, while our use of the MD norm is the natural quantum analogue of the Dobrushin norm on probability measures, we work concretely in the trace-class framework, identifying states with positive trace-class operators of unit trace and measuring distinguishability via the trace norm. We do not attempt here to formulate the theory in a fully general algebraic setting of arbitrary von Neumann algebras and their preduals, where additional subtleties related to modular theory and the geometry of non-commutative \(L^p\)-spaces arise \cite{ Yeadon1980, Marrakchi2024}. In particular, our arguments rely crucially on finiteness of the dimension in several places (e.g., compactness and norm equivalences) and do not immediately extend to infinite-dimensional \(\mathcal{H}\).

\medskip

These limitations suggest several directions for further research, both mathematically and from the perspective of quantum information theory. A natural first step is to extend the  present framework from finite-dimensional matrix algebras to infinite-dimensional systems at the level of trace-class operators on a separable Hilbert space \(\mathcal{H}\), so that channels act on \(\mathcal{T}_1(\mathcal{H})\) and the MD norm remains concretely identified with the trace norm. This would already cover many models of infinite spin systems within the trace-class picture, while avoiding the full technical overhead of general von Neumann algebraic machinery \cite{Ohya1981, Bowen2021}.   Even within the finite-dimensional and trace-class regimes, several conceptual questions remain open.  While  the present framework yields a flexible, quantitative mixing criterion, it is currently unclear whether it is equivalent to, strictly stronger than, or genuinely independent of the irreducibility assumptions typically used to guarantee uniqueness and faithfulness of invariant states or strict positivity of iterates of a quantum channel. Clarifying this relationship would help organize different notions of quantum mixing into a more refined hierarchy, tailored to the structural features (e.g., symmetries, conservation laws, locality) of concrete quantum information processing tasks.

\medskip

From the perspective of many-body quantum information, it is natural to ask how qMD-type conditions behave in more complex geometries and in the presence of disorder. One promising direction is the extension of the present analysis to multidimensional and networked settings, where one considers families of channels \(\{\Phi_e\}\) indexed by edges or regions of a graph, tree, or higher-dimensional lattice, and the ergodic behaviour reflects a non-trivial interplay between temporal inhomogeneity and spatial structure \cite{Be06, EcRu85, monteiro2021quantum}. Understanding how MD-type bounds propagate along such geometries, how they depend on boundary conditions, and how they constrain correlation decay and phase structure in spin-chain and spin-lattice models would yield a more complete picture of non-stationary quantum many-body dynamics. Finally, introducing randomness at the level of the channel sequence---for instance via random quantum channels \cite{kukulski2021generating, matsoukas2024quantum, NP25}, random circuit models, or random MPS---suggests a probabilistic variant of the qMD approach in which one formulates almost-sure or typical mixing criteria in terms of random Markov--Dobrushin data. This would connect the present framework to random operator theory and concentration of measure techniques, and could lead to MD-based conditions for almost-sure mixing and typical ergodic behaviour in noisy quantum information processing architectures.

\medskip

In summary, the perspective developed here provides a flexible and computationally accessible method for controlling mixing and ergodicity in time-inhomogeneous quantum systems, phrased directly in terms of quantum channels and their contraction properties in trace distance. It offers time-uniform control of long-time behaviour together with concrete applications to inhomogeneous tensor-network states, while at the same time making explicit its limitations: the Markov--Dobrushin contraction coefficients are generally non-optimal, and the present theory is confined to finite-dimensional, trace-class settings. Both aspects point towards natural extensions: sharper contraction theories that integrate spectral and functional-inequality techniques, and a fully-fledged infinite-dimensional, trace-class (and eventually von Neumann algebraic) generalization adapted to the broad range of models arising in quantum information theory and quantum many-body physics.

\section*{Conflict of Interests / Competing Interests}

The author declares that there are no conflicts of interest or competing interests associated with this work.

\section*{Data Availability}
No data was used to support the findings of this study.


\begin{thebibliography}{99}


\bibitem{AccLuSou22} Accardi L., Lu Y.G., Souissi A.,  A Markov-Dobrushin inequality for quantum channels, \textit{Open Systems \& Information Dynamics}, Vol.. 28, No. 04, 2150018 (2021)

\bibitem{Acc75} Accardi, L.  On the noncommutative Markov property. Functional Analysis and Its Applications, 9(1), 1-8, (1975).

\bibitem{AccOh99} Accardi, L. and Ohya, M., Compound channels, transition expectations, and liftings. \textit{Applied Mathematics and Optimization}, 39, pp.33-59 ( 1999).

\bibitem{ABM24} Aravinda, S., Banerjee, S., Modak, R.,  Ergodic and mixing quantum channels: From two-qubit to many-body quantum systems. \textit{Physical Review A}, 110(4), p.042607 (2024).

\bibitem{AFk23} Amato, D., Facchi, P. and Konderak, A.,  Asymptotics of quantum channels. Journal of Physics A: Mathematical and Theoretical, 56(26), p.265304 (2023).

\bibitem{Aci} Acín, A. and Masanes, L.,  Certified randomness in quantum physics. Nature, 540(7632), pp.213-219, (2016).

\bibitem{albert2019asymptotics}
Victor V. Albert, Asymptotics of quantum channels: conserved quantities, an adiabatic limit, and matrix product states, \textit{Quantum}, 3:151, 2019.

\bibitem{amato2023asymptotics}
  Amato D.,  Facchi P.,    Konderak A., Asymptotics of quantum channels, \textit{Journal of Physics A: Mathematical and Theoretical}, 56(26):265304, 2023.

\bibitem{aravinda2024ergodic}
Aravinda, S., Banerjee, S., Modak, R. Ergodic and mixing quantum channels: From two-qubit to many-body quantum systems. \textit{Phys. Rev. A} \textbf{110}, 042607 (2024).

\bibitem{ASS20} Accardi, L., Souissi, A.,  Soueidy, E. G. Quantum Markov chains: A unification approach. Infinite Dimensional Analysis, Quantum Probability and Related Topics, 23(02), 2050016, (2020).

\bibitem{bau2013} Burgarth, D., Chiribella, G., Giovannetti, V., Perinotti, P. and Yuasa, K.,  Ergodic and mixing quantum channels in finite dimensions. New Journal of Physics, 15(7), p.073045 (2013).
\bibitem{Ohya1981} Ohya M., Quantum ergodic channels in operator algebras, \emph{J. Math. Anal. Appl.} \textbf{84}(2):318--327 (1981).


\bibitem{singh2024zero}
Singh S., Rahaman M., Datta N., Zero-error communication under discrete-time Markovian dynamics, \emph{Quantum} \textbf{9}:1910 (2025).
\bibitem{singh2024ergodic}
Singh, S., Datta, N., Nechita, I. Ergodic theory of diagonal orthogonal covariant quantum channels. \textit{J. Stat. Phys.} \textbf{114}, 121 (2024).


\bibitem{V24} Vassiliou, P.C.,  Strong Ergodicity in Nonhomogeneous Markov Systems with Chronological Order. \textit{Mathematics}, 12(5), p.660 (2024).

 
\bibitem{Bowen2021} Bowen L., Hayes B., Lin Y.F., A multiplicative ergodic theorem for von Neumann algebra valued cocycles, \emph{Commun. Math. Phys.} \textbf{384}(2):1291--1350 (2021).

\bibitem{Yeadon1980} Yeadon F.J., Ergodic theorems for semifinite von Neumann algebras: II, \emph{Math. Proc. Camb. Philos. Soc.} \textbf{88}(1):135--147 (1980).
\bibitem{Marrakchi2024} Marrakchi A., Vaes S., Ergodic states on type III$_1$ factors and ergodic actions, \emph{J. Reine Angew. Math.} \textbf{809}:247--260 (2024).

\bibitem{belov2025quantum} Belov, M.G., Dubov, V.V., Filimonov, A.V., Malyshkin, V.G. Quantum channel learning. \textit{Phys. Rev. Res.} \textbf{111}, 015302 (2025).

\bibitem{Hu96} Huang, C. C., Isaacson, D.,   Vinograde, B. (1976). The rate of convergence of certain nonhomogeneous Markov chains. Zeitschrift für Wahrscheinlichkeitstheorie und Verwandte Gebiete, 35(2), 141-146.

\bibitem{ben2005} Bengtsson, I., Ericsson, Å., Kuś, M., Tadej, W. and Życzkowski, K., Birkhoff's polytope and unistochastic matrices, N= 3 and N= 4. \textit{Communications in mathematical physics}, 259, pp.307-324 ( 2005).
\bibitem{Be06} Berkovitz, J., Frigg, R.,  Kronz, F., The ergodic hierarchy, randomness and Hamiltonian chaos. Studies in History and Philosophy of Science Part B: Studies in History and Philosophy of Modern Physics, 37(4), 661-691,  (2006).

\bibitem{Birkoff31}  Birkhoff, George D. Proof of the ergodic theorem. Proceedings of the National Academy of Sciences 17.12 (1931): 656-660.

\bibitem{blume2008characterizing}
Robin Blume-Kohout, Hui Khoon Ng, David Poulin, and Lorenza Viola, Characterizing the structure of preserved information in quantum processes, \textit{Physical Review Letters}, 100(3):030501, 2008.

\bibitem{blume2010information}
Robin Blume-Kohout, Hui Khoon Ng, David Poulin, and Lorenza Viola, Information-preserving structures: A general framework for quantum zero-error information, \textit{Physical Review A}, 82(6):062306, 2010.

\bibitem{Bolzmann} Boltzmann, L., Einige allgemeine Sätze über Wärmegleichgewicht: vorgelegt in der Sitzung am 13,   K. und k. Hof-und Staatsdr, (1871)


\bibitem{BDX06} Boyd, S., Diaconis, P., Sun, J.,  Xiao, L.  Fastest mixing Markov chain on a path. The American Mathematical Monthly, 113(1), 70-74 (2006).

\bibitem{L22} Longla, M., Mous-Abou, H.,   Ngongo, I. S.  On some mixing properties of copula-based Markov chains. Journal of Statistical Theory and Applications, 21(3), 131-154 (2022).

\bibitem{MixOQS23} Bravyi, S., Carleo, G., Gosset, D.,  Liu, Y.  A rapidly mixing Markov chain from any gapped quantum many-body system. Quantum, 7, 1173 (2023).


\bibitem{BR} Bratteli, O. and Robinson, D.W., 2012. Operator algebras and quantum statistical mechanics: Volume 1: C$^*$-and W$^*$-Algebras. Symmetry Groups. Decomposition of States. Springer Science \& Business Media.

\bibitem{brasil2021lyapunov}
Brasil, J.E., Knorst, J., Lopes, A.O. Lyapunov exponents for quantum channels: An entropy formula and generic properties. \textit{J. Math. Phys.} \textbf{19}, 155-187 (2021).

\bibitem{Breamaud20} Brémaud, P., Non-homogeneous Markov chains. In Markov Chains: Gibbs Fields, Monte Carlo Simulation and Queues Cham: Springer International Publishing, pp. 399-422, (2020).

\bibitem{George2025} George I., Hirche C., Nuradha T., Wilde M.M., Quantum Doeblin coefficients: Interpretations and applications, Preprint \texttt{arXiv:2503.22823} (2025).
\bibitem{Delsol2025} Delsol I., Fawzi O., Kochanowski J., Ramachandran A., Computational aspects of the trace norm contraction coefficient, Preprint \texttt{arXiv:2507.16737} (2025).
\bibitem{Hirche2024} Hirche C., Quantum Doeblin coefficients: A simple upper bound on contraction coefficients, In: \emph{Proc. IEEE Int. Symp. Inf. Theory (ISIT)}, pp. 557--562, IEEE (2024).
\bibitem{Carbonna20} Carbone, Raffaella, and Anna Jenčová, On period, cycles and fixed points of a quantum channel, Annales Henri Poincaré. Vol. 21. No. 1. Cham: Springer International Publishing, 2020.

\bibitem{CFS82} I. P. Cornfeld, S. V. Fomin, and Ya. G. Sinai, Ergodic Theory, Vol. 245 (Springer, 1982).

\bibitem{Choi75} Choi, M.D., 1975. Completely positive linear maps on complex matrices. Linear algebra and its applications, 10(3), pp.285-290.
\bibitem{Chen24} Chen, L., Garcia, R. J., Bu, K.,  Jaffe, A., Magic of random matrix product states. Physical Review B, 109(17), 174207,  (2024).
\bibitem{CN10} Collins, B., Nechita, I.: Random quantum channels I: Graphical calculus and the Bell state phenomenon. Commun. Math. Phys. 297(2), 345-370 (2010)

\bibitem{CN11II} Collins, B., Nechita, I.: Random quantum channels II: Entanglement of random subspaces, Renyi entropy estimates and additivity problems. Adv. in Math. 226(2), 1181-1201 (2011).

\bibitem{D56} R. L. Dobrushin, Central limit theorem for
nonstationary Markov chains. I,II, \textit{Theor. Probab. Appl.}
{\bf 1}(1956),65-80; 329-383.

\bibitem{EcRu85} Eckmann, J. P.,  Ruelle, D.  Ergodic theory of chaos and strange attractors. Reviews of modern physics, 57(3), 617, (1985).

\bibitem{Ehrenfest1907} Ehrenfest, P. U. T. Begriffliche Grundlagen der statistischen Auffassung in der Mechanik. Vieweg- Teubner Verlag, (1907).


\bibitem{fang2021geometric}
Fang, K., Fawzi, H. Geometric Rényi divergence and its applications in quantum channel capacities. \textit{Commun. Math. Phys.} \textbf{384}, 1615-1677 (2021).

\bibitem{Fid09} Fidaleo,F., On strong ergodic properties of quantum dynamical systems. Infin. Dimens. Anal. Quantum  Probab. Relat. Top. 12 (04), 551-564 (2009).

\bibitem{Fid10} Fidelio, Francesco. The entangled ergodic theorem in the almost periodic case. Linear algebra and its applications,  vol. 432, no 2-3, p. 526-535,(2010).


\bibitem{GOZ10}  S. Garnerone, T. R. de Oliveira, S. Haas, and P. Zanardi,
Statistical properties of random matrix product states, Phys.
Rev. A 82, 052312 (2010).

\bibitem{gour2021entropy}
Gour, G., Wilde, M.M. Entropy of a quantum channel. \textit{PRX Quantum} \textbf{3}, 023096 (2021).

\bibitem{gyongyosi2018survey}
Gyongyosi, L., Imre, S., Nguyen, H.V. A survey on quantum channel capacities. \textit{Quantum Inf. Process.} \textbf{20}, 1149-1205 (2018).

 \bibitem{GQ15} Gaubert, S.,  Qu, Z.,  Dobrushin’s ergodicity coefficient for Markov operators on cones. Integral Equations and Operator Theory, 81(1), 127-150 (2015).

\bibitem{han2022quantum}
Han, Q., Han, Y., Kou, Y., Bai, N. Quantum channel measurement with local quantum Bernoulli noises. \textit{Sci. Rep.} \textbf{12}, 12929 (2022).

\bibitem{Haj58}  Hajnal, J., Bartlett, M.S.: Weak ergodicity in non-homogeneous Markov chains. Math. Proc. Camb.  Philos. Soc. 54(2), 233–246 (1958).

 \bibitem{Ib97} Ibragimov, I. D. A., Dobrushin's works on Markov processes. Russian Mathematical Surveys, 52(2), 239, (1997).

\bibitem{Kraus71} K Kraus. General state changes in quantum theory. Annals of Physics, 64:311-335, June (1971).

\bibitem{KR50}  Krein, M. G. and Rutman, M. A., Linear operators leaving invariant a cone in a Banach space, Am. Math. Soc. Transl. (26), 128, (1950)

\bibitem{kukulski2021generating}
Kukulski, R., Nechita, I., Pawela, Ł., Puchała, Z., Życzkowski, K. Generating random quantum channels. \textit{J. Phys. A: Math. Theor.} \textbf{62}, 045301 (2021).

\bibitem{linden2022arbitrary}
Noah Linden and Paul Skrzypczyk, How to use arbitrary measuring devices to perform almost perfect measurements, \textit{arXiv preprint} arXiv:2203.02593, 2022.

\bibitem{lloyd1997capacity}
Lloyd, S. Capacity of the noisy quantum channel. \textit{Phys. Rev. A} \textbf{55}, 1613-1622 (1997).

\bibitem{ma2023sequential}
Wen-Long Ma, Shu-Shen Li, and Ren-Bao Liu, Sequential generalized measurements: Asymptotics, typicality, and emergent projective measurements, \textit{Physical Review A}, 107(1):012217, 2023.

\bibitem{markov06}
A. A. Markov,
 Rasprostranenie zakona bol'shikh chisel na velichiny, zavisyashchie drug ot druga, (Extension of the law of large numbers to quantities depending on each other),
Izvestiya Akademii Nauk SPb, Series 6, \textbf{19} (1906), 61-80.

\bibitem{Bau76} Bowerman, B., David, H. T.,  Isaacson, D. The convergence of Cesaro averages for certain nonstationary Markov chains. Stochastic processes and their applications, 5(3), 221-230  (1977).


\bibitem{matsoukas2024quantum}
Matsoukas-Roubeas, A.S., Prosen, T., del Campo, A. Quantum chaos and coherence: Random parametric quantum channels. \textit{Quantum Sci. Technol.} \textbf{8}, 1446 (2024).
 

\bibitem{monteiro2021quantum}
Monteiro, F., Tezuka, M., Altland, A., Huse, D.A., Micklitz, T. Quantum ergodicity in the many-body localization problem. \textit{Phys. Rev. Lett.} \textbf{127}, 030601 (2021).

\bibitem{Moore15} C. C. Moore, Ergodic Theorem, Ergodic Theory, and Statistical Mechanics, Proc. Natl. Acad. Sci. U.S.A. 112, 1907 (2015).

\bibitem{movassagh2021theory}
Movassagh, R., Schenker, J. Theory of ergodic quantum processes. \textit{Phys. Rev. X} \textbf{11}, 041001 (2021).

\bibitem{movassagh2022a}
Movassagh, R., Schenker, J. An ergodic theorem for quantum processes with applications to matrix product states. \textit{Commun. Math. Phys.} \textbf{395}, 1175-1196 (2022).

\bibitem{Muk15} Mukhamedov, F. (2015). Ergodic properties of nonhomogeneous Markov chains defined on ordered Banach spaces with a base. Acta Mathematica Hungarica, 147(2), 294-323.

\bibitem{V22} Vassiliou, P. C. G.  Laws of large numbers for non-homogeneous Markov systems with arbitrary transition probability matrices. Journal of Statistical Theory and Practice, 16(2), 18 (2022).

\bibitem{mukhamedov2015ergodic}
Mukhamedov, F. Ergodic properties of nonhomogeneous Markov chains defined on ordered Banach spaces with a base. \textit{J. Math. Anal. Appl.} \textbf{147}, 294-323 (2015).

\bibitem{mukhamedov2021approximations}
Mukhamedov, F., Al-Rawashdeh, A. Approximations of non-homogeneous Markov chains on abstract states spaces. \textit{Stoch. Anal. Appl.} \textbf{11}, 2150002 (2021).

\bibitem{mukhamedov2022generalized}
Mukhamedov, F., Al-Rawashdeh, A. Generalized Dobrushin ergodicity coefficient and ergodicities of non-homogeneous Markov chains. \textit{Stat. Probab. Lett.} \textbf{16}, 18 (2022).

\bibitem{MG19} Mukhamedov, F.,  El Gheteb, S.  Clustering property of Quantum Markov Chain associated to XY-model with competing Ising interactions on the Cayley tree of order two. Mathematical Physics, Analysis and Geometry, 22, 1-15, (2019).

\bibitem{MuWa2018} Mukhamedov, F. and Watanabe, N., On S-mixing entropy of quantum channels. Quantum Information Processing, 17, pp.1-21 (2018).

\bibitem{nelson2024ergodic}
Nelson, B., Roon, E.B. Ergodic quantum processes on finite von Neumann algebras. \textit{J. Funct. Anal.} \textbf{287}, 110485 (2024).

\bibitem{NP12}  Nechita I.,   Pellegrini C., Random repeated quantum interactions and random in variant states. Probab. Theory Relat. Fields, 152 (1-2):299-320,  (2012).

\bibitem{NP25} Nechita, I., Park, S.J., 2025, March. Random covariant quantum channels. In Annales Henri Poincaré (pp. 1-61). Springer International Publishing.

\bibitem{perez2007matrix} Perez-Garcia D, Verstraete F, Wolf M M,  Cirac J I,   Matrix product state representations arXiv preprint quant-ph/0608197, (2006).

\bibitem{raginsky2002strictly}
Maxim Raginsky, Strictly contractive quantum channels and physically realizable quantum computers, \textit{Physical Review A}, 65(3):032306, 2002.

\bibitem{regula2021fundamental}
Regula, B., Takagi, R. Fundamental limitations on distillation of quantum channel resources. \textit{Nat. Commun.} \textbf{12}, 4411 (2021).

\bibitem{SB25} Souissi, A., Barhoumi, A. An Exponential Mixing Condition for Quantum Channels: Application to Matrix Product States. Quantum Inf Process 24, 150 (2025). https://doi.org/10.1007/s11128-025-04762-1.

\bibitem{SA26} Souissi, A., Andolsi, A. Infinite-scale decoherence of GGHZ-states: an algebraic approach. Quantum Inf Process 25, 108 (2026). https://doi.org/10.1007/s11128-026-05131-2 
    
\bibitem{Sa15} Saburov, M. (2016). Ergodicity of nonlinear Markov operators on the finite dimensional space. Nonlinear Analysis: Theory, Methods \& Applications, 143, 105-119.
\bibitem{singh2022detecting}
Singh, S., Datta, N. Detecting positive quantum capacities of quantum channels. \textit{Quantum} \textbf{8}, 50 (2022).

\bibitem{SM24} Souissi, A.,  Mukhamedov, F., Nonlinear Stochastic Operators and Associated Inhomogeneous Entangled Quantum Markov Chains. Journal of Nonlinear Mathematical Physics, 31(1), 11, (2024).

\bibitem{SSB23} Souissi, A., Soueidi E.L., Barhoumi, A.,On a $\psi$-mixing property for entangled Markov chains. Physica A: Statistical Mechanics and its Applications, 613, 12853, (2023).
 
\bibitem{verstraete2008matrix} Verstraete F, Murg V and Cirac J I,  Matrix product states, projected entangled pair states, and variational renormalization group methods for quantum spin systems Advances in physics 57(2) 143-224, (2008).

\bibitem{vonN32}   von Neumann, J.: Proof of the quasi-ergodic hypothesis. Proc. Natl. Acad. Sci. 18(1), 70-82 (1932)
 

\bibitem{V25} Veretennikov, A.,   Nurieva, A. (2025). On conditions for Dobrushin's Central limit theorem for non-homogeneous Markov chains. arXiv preprint arXiv:2506.07287.

\bibitem{W12} Wolf, M.M.: Quantum channels and operations: Guided tour (unpublished) (2012)

\end{thebibliography}
\end{document}